\documentclass[prx,twocolumn,floatfix,notitlepage,groupedaddress,longbibliography, nofootinbib]{revtex4-1}

\usepackage[colorlinks=true,citecolor=blue,linkcolor=magenta]{hyperref}
\usepackage{mathtools}

\usepackage[dvipsnames]{xcolor}
\usepackage{amssymb,amsmath,amstext,amsthm}
\usepackage{graphicx}
\usepackage{epstopdf}
\usepackage{color}
\usepackage{bm}
\usepackage{appendix}
\usepackage[T1]{fontenc}
\usepackage{bbm}
\usepackage{latexsym}

\usepackage{float}
\usepackage{verbatim}
\usepackage{lipsum}
\usepackage{braket}

\newcommand{\del}[0]{\partial}

\renewcommand{\Re}{\text{Re}}
\renewcommand{\Im}{\text{Im}}

\newtheorem{theorem}{Theorem}
\newtheorem{lemma}{Lemma}
\newtheorem{corollary}{Corollary}

\newcommand{\SC}{\mathcal{S}}

\newcommand{\Tr}{{\rm Tr}}

\renewcommand{\geq}{\geqslant}

\renewcommand{\Re}{\text{Re}}
\renewcommand{\Im}{\text{Im}}

\newcommand{\bramatket}[3]{\langle #1 \hspace{1pt} | #2 | \hspace{1pt} #3 \rangle}
\newcommand{\bramatketq}[2]{\bramatket{#1}{#2}{#1}}

\renewcommand{\vec}[1]{\boldsymbol{#1}}  

\newcommand*{\id}{\openone}

\usepackage[linesnumbered,ruled,vlined]{algorithm2e}
\SetKwInput{kwInit}{Init}

\begin{document}

\title{Universal Compiling and (No-)Free-Lunch Theorems\\ for Continuous Variable Quantum Learning}

\author{Tyler Volkoff}
\affiliation{Theoretical Division, Los Alamos National Laboratory, Los Alamos, NM, USA}
\author{Zo\"{e} Holmes}
\affiliation{Information Sciences, Los Alamos National Laboratory, Los Alamos, NM, USA}
\author{Andrew Sornborger}
\affiliation{Information Sciences, Los Alamos National Laboratory, Los Alamos, NM, USA}
\affiliation{Quantum Science Center, Oak Ridge, TN 37931, USA}

\begin{abstract}
Quantum compiling, where a parameterized quantum circuit is trained to learn a target unitary, is an important primitive for quantum computing that can be used as a subroutine to obtain optimal circuits or as a tomographic tool to study the dynamics of an experimental system. While much attention has been paid to quantum compiling on discrete variable hardware, less has been paid to compiling in the continuous variable paradigm. Here we motivate several, closely related, short depth continuous variable algorithms for quantum compilation. We analyse the trainability of our proposed cost functions and numerically demonstrate our algorithms by learning arbitrary Gaussian operations and Kerr non-linearities. We further make connections between this framework and quantum learning theory in the continuous variable setting by deriving No-Free-Lunch theorems. These generalization bounds demonstrate a linear resource reduction for learning Gaussian unitaries using entangled coherent-Fock states and an exponential resource reduction for learning arbitrary unitaries using Two-Mode-Squeezed states.
\end{abstract}
\maketitle

\section{Introduction}\label{sc:intro}

Progress in experimental implementations of quantum optical neural networks~\cite{engreview,engdeep,engproc} and extensions of quantum machine learning frameworks to the continuous-variable (CV) setting~\cite{PhysRevLett.118.080501, Miatto2020fastoptimizationof, PhysRevA.102.012417} indicate that quantum photonics is a viable platform for near-term quantum algorithms.
Variational quantum algorithms, 
where a problem-specific cost function is evaluated on a quantum computer~\cite{cerezo2020variationalreview, bharti2021noisy}, while a classical optimiser trains a parameterized quantum circuit to minimise this cost, have been implemented in photonic systems.  
For instance, the variational quantum eigensolver~\cite{peruzzo2014VQE} and variational quantum unsampling~\cite{engunsamp}, i.e., partial characterization of a unitary operator, have both been implemented on integrated photonic processors.
Beyond the fundamental physical advantages of photonic systems, such as a well-characterized set of loss channels and the possibility of room temperature operation, there are computational advantages to CV implementations of variational quantum algorithms such as the existence of efficient quantum error mitigation schemes~\cite{PhysRevA.97.032346,volkhier, su2020error}. 

An important computational task that CV quantum processors are well-suited to is the variational compilation~\cite{Khatri2019quantumassisted,arrazolaml, sharma2019noise} of CV unitaries. 
The task is to optimize a parameterized quantum circuit to learn a given target unitary.  
The target unitary could take the form of a known gate sequence that one seeks to compile into a shorter depth, or more noise resistant, circuit. Hence quantum compiling could be used as a subroutine to reduce the resources required to implement large scale quantum algorithms. 
Alternatively, the target unitary could be the unknown dynamics of a quantum system. In this case, quantum compilation plays a role analogous to, but potentially less resource intensive than, a quantum sensing protocol~\cite{Degen2017QuantumSensing} or unitary process tomography~\cite{Gus2014Process,Baldwin2014Quantum}. Specifically, our CV compiling algorithms make use of Gaussian measurements and CV resources such as intensity and quadrature squeezing, and so do not require preparation of exotic optimal probe states as in an  optimal quantum sensing protocol, nor a large number of measured observables as in process tomography. In this sense, CV compiling provides a new tool for experimental physics. 

In this paper, we establish frameworks for CV variational quantum compiling that are valid for arbitrary CV target unitaries. In contrast to the variational compiling method explored in Ref.~\cite{arrazolaml}, we include entanglement-enhanced methods that can be used to learn an entire unitary rather than just its action on a low lying subspace. We illustrate the wide applicability of our cost functions for CV quantum compiling by numerically demonstrating efficient learning of arbitrary single-mode Gaussian unitaries, the generalized beamsplitter operation, and Kerr non-linearities.

We further make connections between this framework and CV quantum learning theory by deriving ``(No)-Free-Lunch'' theorems. These analytic theorems specify the minimal training data required to learn CV unitary operators in increasingly general settings, providing fundamental bounds on the limits of quantum learning. In particular, the bounds highlight how utilizing entangled training states can reduce the amount of training data required to learn an unknown unitary and thus entanglement could be seen to provide a `free lunch'. We further use these results as an alternative motivation for the cost functions we propose for quantum compiling.

This manuscript is structured as follows. Section~\ref{sec:back} provides a background to quantum compiling, including a discussion of its possible uses and a summary of previously proposed methods for discrete variable quantum compilation. Section~\ref{sec:algs} presents our main results, including the cost functions we propose for CV quantum compiling and an analysis of their trainability. Section~\ref{sec:gng} contains numerical implementations of our proposed CV learning algorithm. Section~\ref{sec:cvnfl} presents our No-Free-Lunch theorems for CV learning.
Section~\ref{sec:disc} summarises and discusses our results.

\section{Background}\label{sec:back}

\subsection{Applications of CV quantum compiling}

The goal of CV quantum compiling is to take a (possibly unknown) unitary $U$ and return a gate sequence $V$, executable on a CV quantum computer, that has approximately the same action as $U$ on any given input state (up to possibly a global phase factor). Here we describe three possible applications of this subroutine. 

\paragraph*{Optimal circuit design.} Quantum compilation could be used to variationally compile CV gate sequences to form optimal subcircuits. By optimal, we primarily mean short depth. However, compilation might also be used to find circuits that naturally compensate for systematic gate errors or that are more generally resistant to noise. The construction of such optimal circuits may prove critical for the successful implementation of larger scale algorithms, including proposals for generating optimal bosonic states in protocols for quantum metrology~\cite{PhysRevA.94.042327} and entanglement extraction~\cite{PhysRevA.100.022331}.

\paragraph*{Experimental quantum physics.} More generally, variational quantum compilation could be used to learn the unknown unitary dynamics of a physical system. In the context of an optical system, one might be interested in studying the optical properties of a new material as sketched in Fig.~\ref{fig:ExperimentalSchematic}. For example, as discussed further in Section~\ref{sec:gng}, one might use quantum compiling to estimate the Kerr effect in nonlinear optical media that cannot itself be directly implemented in a CV quantum circuit. In this manner, variational quantum compilation provides a new tool for experimental physics.

\begin{figure}
    \centering
    \includegraphics[width=0.49\textwidth]{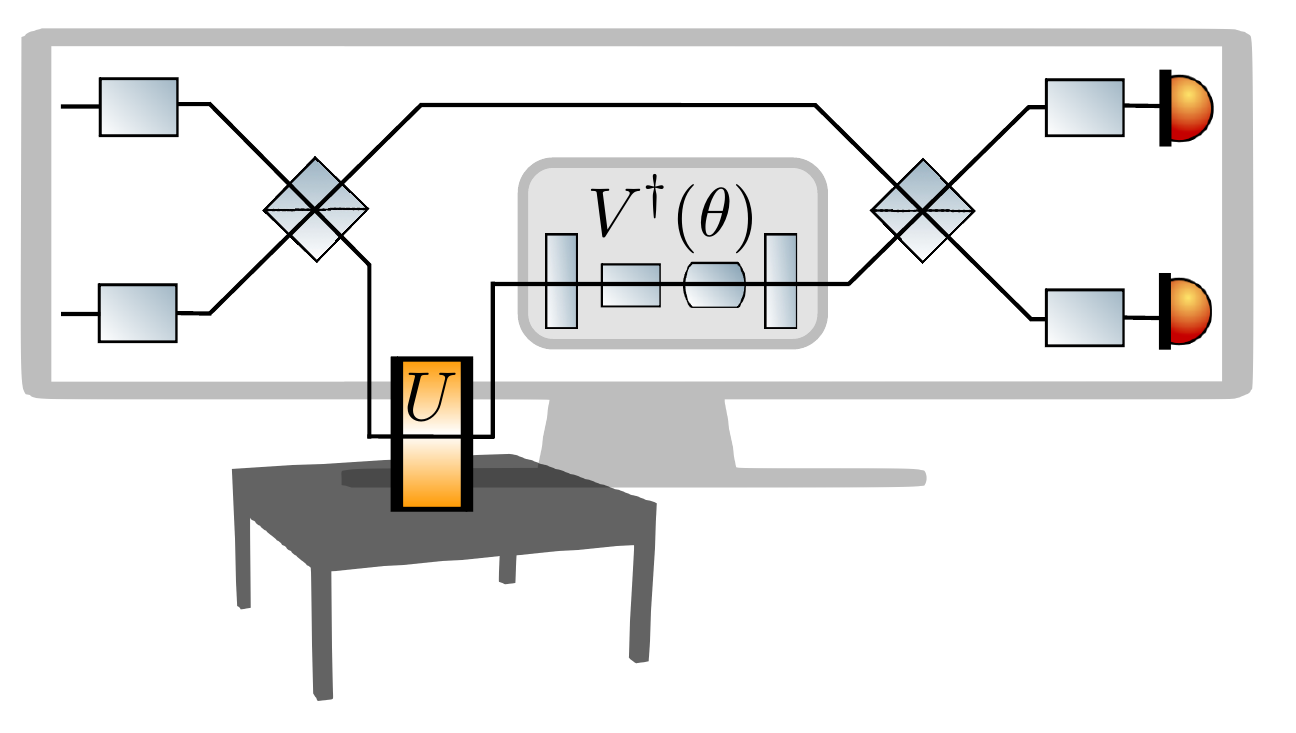}
    \caption{\textbf{Learning experimental CV quantum systems.} Here we sketch an experimental circuit to learn the unitary $U$ implemented by a novel optical material (shown in orange) by training a parameterised quantum circuit $V^\dagger(\theta)$ implemented on an optical quantum computer (shown in grey). For specific details on the proposed circuit, and the cost it calculates, see Fig.~\ref{fig:CostFuncs}(d).}
    \label{fig:ExperimentalSchematic}
\end{figure}

\paragraph*{Structured learning.}
In discrete variable systems, variational quantum compilation has proven useful for learning the spectral decomposition of a unitary operation. This in turn opens up the possibility of simulating beyond the coherence time of a quantum processor~\cite{vff, gibbs2021long, geller2021experimental}. Similarly one could use discrete variable quantum compiling to learn block decompositions of a given unitary which is useful to study the entanglement properties of a system. It would be interesting to explore whether CV quantum compiling could similarly be used to study the spectral or entanglement properties of a given CV unitary, or for simulating the dynamics of CV quantum systems~\cite{PhysRevA.97.062311,PhysRevA.99.022341,bhaskar}.

\subsection{Discrete Variable Quantum Compilation}

Before presenting our algorithm for continuous variable quantum compilation, let us first review the discrete variable Quantum Assisted Quantum Compilation (QAQC) algorithm of Ref.~\cite{Khatri2019quantumassisted}. In QAQC a compilation is found by variationally searching for a gate sequence that minimizes the Hilbert-Schmidt Test cost. This cost, which quantifies how close the compilation is to exact, can be written as the normalized Hilbert-Schmidt inner product between the target unitary $U$ and possible compilation $V$,
\begin{equation}\label{eq:HST}
    C_{\text{HST}}(V,U):= 1-{1\over d^{2}}\vert \text{Tr} \left(V^{\dagger}  U  \right) \vert^{2} \, .
\end{equation}
This cost is faithful, vanishing if and only if $U$ and $V$ differ by a global phase factor, i.e., $V=e^{i\varphi}U$ for some $\varphi\in\mathbb{R}$. Therefore, by minimizing $C_{\text{HST}}$, we learn a unitary $U$ that implements a target $V$ up to a global phase.

The Hilbert-Schmidt Test cost may be computed by the two closely related circuits shown in Fig.~\ref{fig:CostFuncs}(a) and Fig.~\ref{fig:CostFuncs}(b). To see how, we first note that
\begin{align}\label{eq-HS_inner_prod} 
\frac{1}{d}\Tr(V^\dagger U) = \bramatketq{\Phi^+}{V^\dagger U\otimes \mathbb{I}}  \, , 
\end{align}
where $\ket{\Phi}_{AB}$ is the Bell entangled state of two qubit registers $A$ and $B$ of $n=\log_{2}d$ qubits, i.e. $\ket{\Phi}_{AB}:=\bigotimes_{j=1}^{n}\ket{\Phi^{+}}_{A_{j}B_{j}}$ with $\ket{\Phi^{+}}:={1\over \sqrt{2}}(\ket{00}+\ket{11})$. 
It thus follows that we can write 
\begin{equation}\label{eq:C_HST_Bell}
    C_{\text{HST}}(V,U)= 1- \vert\bramatketq{\Phi^+}{V^\dagger U \otimes \mathbb{I}}\vert^{2} \, 
\end{equation}
and $C_{\text{HST}}$ can be computed using the circuit shown in Fig.~\ref{fig:CostFuncs}(a).  Due to the ricochet property of the state $\ket{\Phi}$, viz., $X\otimes \mathbb{I}\ket{\Phi^{+}}=\mathbb{I}\otimes X^{T}\ket{\Phi^{+}}$
for linear operator $X$, the Hilbert-Schmidt test cost can alternately be written as
\begin{equation}\label{eq:C_HST_Bell_ric}
    C_{\text{HST}}(V,U)= 1- \vert\bramatketq{\Phi^+}{U\otimes V^*}\vert^{2} \, .
\end{equation}
Thus $C_{\text{HST}}(V,U)$ can also be computed with the target and ansatz unitaries applied in parallel, instead of in series, reducing the total circuit depth as shown in Fig.~\ref{fig:CostFuncs}(b).

\medskip
Finally, we note that the Hilbert Schmidt Test cost can be related to the average gate fidelity between $U$ and $V$. Specifically, it can be shown \cite{HHH99,nielsen02} that
\begin{equation}
\label{eq:HST_Fbar}
C_{\text{HST}}(U,V) = \frac{d+1}{d}\left(1-\overline{F}(U,V)\right)\,
\end{equation}
where 
\begin{equation}
    \overline{F}(U,V)\coloneqq\int_{\psi}| \bra{\psi} V^\dagger U\ket{\psi}|^2~\text{d}\psi\,
\end{equation}
is the average fidelity of states acted
upon by $V$ versus those acted upon by $U$, with the
average being over all pure states according to the Haar measure. In theory, Eq.~\eqref{eq:HST_Fbar} provides a third way of measuring $C_{\text{HST}}$. One could perform a Loschmidt echo test, as shown in Fig.~\ref{fig:CostFuncs}(c), using different input states that are sampled according to the Haar measure. However, in practice, this is not a viable training technique, since, as we will discuss in Section~\ref{sec:tmssnfl}, in order to fully learn $U$, the average would need to be taken over an exponentially large number of training states. Instead, the significance of Eq.~\eqref{eq:HST_Fbar} lies in the fact that it embues $C_{\text{HST}}$ with operational meaning for non-zero values since it entails that low cost values correspond to high average gate fidelities.

\begin{figure*}[t!]
    \includegraphics[width=0.95\textwidth]{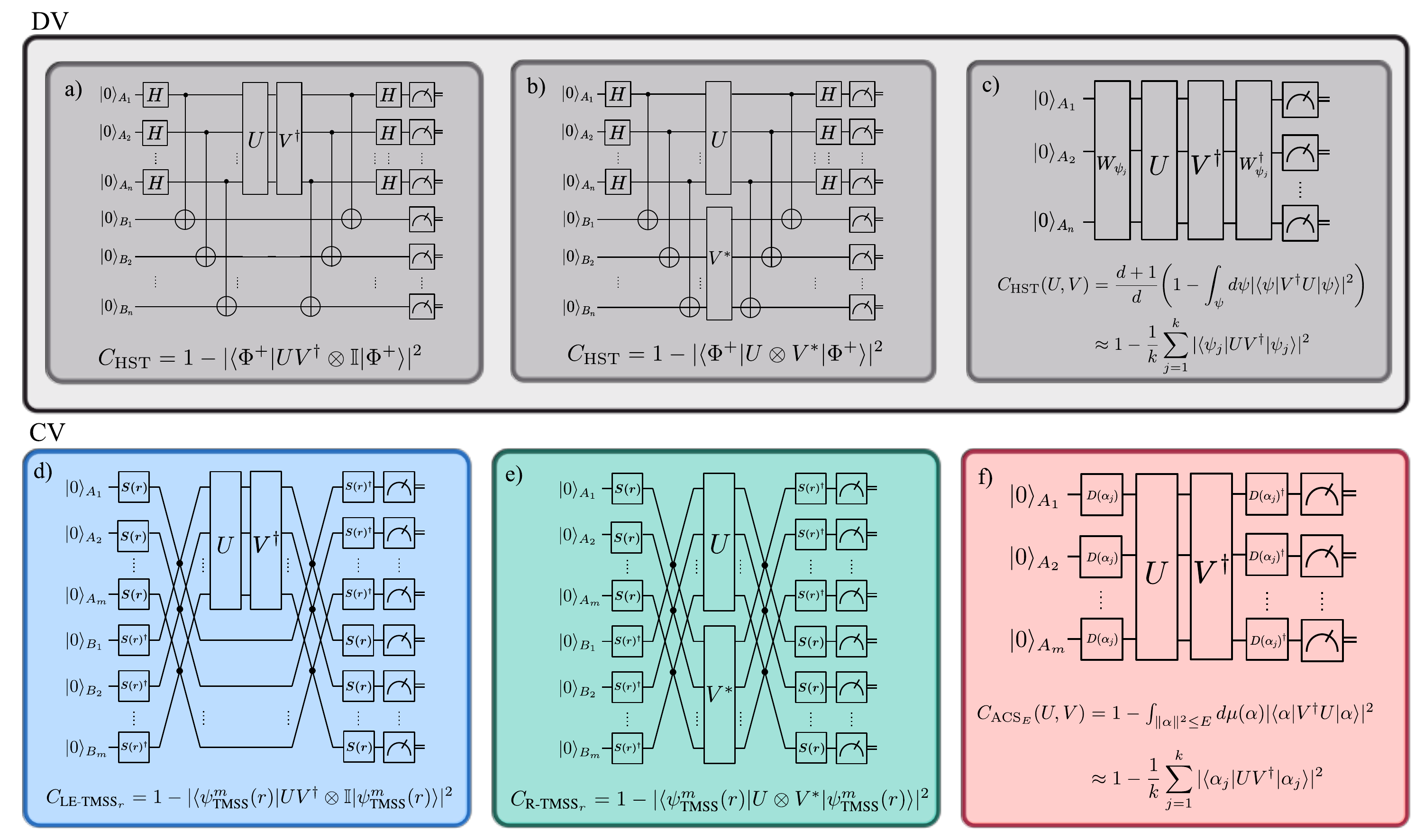}
    \caption{\textbf{Schematic of cost functions.} In this figure we show how the three different cost functions we propose for CV quantum compiling (d - f) are related to three different ways of measuring the Hilbert-Schmidt Test cost for DV compilation (a - c). The probability to measure the all-zero state on all $2n$ qubits in a) and b) is equal to $\Tr[U V^\dagger]/d^2$ and hence can be used to compute $C_{\rm HST}$. In c) $W_{\psi_j}$ is the unitary that prepares the state $\ket{\psi_j}$, i.e. $W_{\psi_j} \ket{0} = \ket{\psi_j}$, therefore the probability to measure the all-zero state on $n$ qubits is equal to $| \bra{\psi_j } U V^\dagger \ket{\psi_j}|^2 $. $C_{\rm HST}$ could theoretically be estimated by running this circuit over a Haar random ensemble of training states but for large problems this is exponentially inefficient. In d) and e)  $S(r)$ is the unitary single mode squeeze operator and $\rangle\!\!\! \bullet \!\!\!\langle$ is a 50:50 beamsplitter that entangles the squeezed registers. The probability to measure the all-zero state on all $2m$ modes in d) and e) is equal to $1 - C_{\rm LE-TMSS}$ and $1 - C_{\rm R-TMSS}$ respectively. In f) $D(\alpha)$ is the unitary single mode displacement operator. The mean probability to obtain all-zero state on $m$ modes is equal to the average of the $k$ values $\vert\bra{\alpha_j } U V^\dagger \ket{\alpha_j}\vert^2 $, which is used to estimate $C_{\rm ACS}$.} 
    \label{fig:CostFuncs}
\end{figure*}

\section{Universal continuous variable quantum compiling}\label{sec:algs}

\subsection{Cost Functions}\label{sec:costs}

For continuous variable quantum compiling, we suppose the target and compiled unitaries $U$ and $V$ act on $m$ CV modes. More concretely, the mathematical setting for CV quantum compiling is the Hilbert space $\mathcal{H}$ of $m$ quantum harmonic oscillators, and the operator algebra $B(\mathcal{H})$ of bounded linear operators. In this infinite dimensional space, the Hilbert-Schmidt inner product used in Eq.~\eqref{eq:HST} is not well-defined and hence cannot be used for CV quantum compilation. However, as indicated in Fig.~\ref{fig:CostFuncs}, we can use generalizations of the three different ways in which Eq.~\eqref{eq:HST} can be implemented, namely via Eq.~\eqref{eq:C_HST_Bell}, Eq.~\eqref{eq:C_HST_Bell_ric} and Eq.~\eqref{eq:HST_Fbar}, to define costs for CV quantum compiling. In contrast to the DV case where Eq.~\eqref{eq:C_HST_Bell}, Eq.~\eqref{eq:C_HST_Bell_ric} and Eq.~\eqref{eq:HST_Fbar} are three equivalent ways to estimate the same cost, here the three costs are fundamentally distinct. 

\medskip

\paragraph*{Loschmidt echo Two-Mode-Squeezed State cost} Let us start by defining a CV generalization of Eq.~\eqref{eq:C_HST_Bell}. To do so, we first note that Two-Mode-Squeezed States (TMSS) are a natural analogue of Bell states for CV systems. The Two-Mode-Squeezed State, acting between two $m$-mode registers $A$ and $B$, is defined as 
\begin{equation}
\begin{aligned}
&\ket{\psi^m_{\text{TMSS}}(r)} := \bigotimes_{j=1}^{m} \ket{\psi_{\text{TMSS}}(r)}_{A_jB_j} \ \ \ \text{with} \\
&\ket{\psi_{\text{TMSS}}(r)}_{A_jB_j} \propto \sum_{n=0}^{\infty}(\tanh r)^{n}\ket{n}_{A_{j}} \otimes \ket{n}_{B_{j}} 
\end{aligned}\label{eqn:tmssa}
\end{equation}
where $\{ \ket{n} \}_{n=0}^\infty$ is the Fock basis and $r$ is a squeezing parameter. To highlight the connection between TMSS and Bell states it is helpful to consider its truncated variant
\begin{align}
\ket{\psi^{\mathfrak{r}}_{\text{TMSS}}(r)}&:=\sqrt{1-\tanh^{2}r\over 1-\tanh^{2\mathfrak{r}}r}\sum_{n=0}^{\mathfrak{r}-1}(\tanh r)^{n}\ket{n}\otimes \ket{n} \, ,
\label{eqn:tmsslim}
\end{align} 
which tends to the standard TMSS in the limit that $\mathfrak{r}$ tends to infinity, i.e. $\lim_{\mathfrak{r}\rightarrow \infty}\ket{\psi^{\mathfrak{r}}_{\text{TMSS}}(r)} =\ket{\psi_{\text{TMSS}}(r)}$. For finite $\mathfrak{r}$ the truncated TMSS tends to a Bell state as $r$ tends to infinity, that is
\begin{align}
\lim_{r\rightarrow \infty}\ket{\psi^{\mathfrak{r}}_{\text{TMSS}}(r)} &= {1\over \sqrt{\mathfrak{r}}}\sum_{n=0}^{\mathfrak{r}-1}\ket{n}\otimes \ket{n} = \ket{\Phi^+} \, .
\end{align}
In this sense, the TMSS may be viewed as a CV generalization of the Bell state.

More generally, TMSSs are highly entangled states which, by reducing the number of measurements necessary to attain a given signal-to-noise ratio, have proven to be an important resource in quantum metrology~\cite{Ligo2013Squeezing,Lawrie2019Quantum, Gessner2020Multiparameter}. Moreover, TMSSs were numerically shown to be nearly optimal for measuring the fidelity of noisy CV quantum teleportation channels~\cite{Sharma2020Opt}. 
These examples suggest that TMSSs may also be valuable for the unitary channel discrimination task we consider here. 
This is confirmed in Section~\ref{sec:cvnfl}, where we use the entanglement-enhanced No-Free-Lunch theorem of Ref.~\cite{nopree} to argue that training on a single TMSS minimizes the generalization error.

This motivates our first proposed cost function to train an $m$-mode hypothesis unitary $V$ to match an $m$-mode target unitary $U$ as the following,
\begin{equation}
    C_{\text{\footnotesize LE-TMSS}_{r}}(V,U):= 1- \vert \bra{\psi^m_{\text{TMSS}}(r)}  UV^\dagger\otimes \id \ket{\psi^m_{\text{TMSS}}(r)} \vert^{2} \, .
    \label{eqn:costtmss}
\end{equation}
This is the CV analogue of Eq.~\eqref{eq:C_HST_Bell} obtained by using an $m$-mode TMSS instead of a Bell state. 
We call this cost, which is evidently faithful by construction, the Loschmidt Echo Two-Mode Squeezed State (LE-TMSS) cost since it measures the inner product between $V \otimes \id \ket{\psi^m_{\text{TMSS}}(r)} $ and $U \otimes \id \ket{\psi^m_{\text{TMSS}}(r)} $ using the Loschmidt Echo circuit sketched in Fig.~\ref{fig:CostFuncs}(d). 

To understand the structure of the circuit that we propose to measure $ C_{\text{\footnotesize LE-TMSS}_{r}}$, it is helpful to recall that the TMSS can be written as 
\begin{equation}
\ket{\psi_{\text{TMSS}}(r)}_{A_jB_j} =e^{{\pi\over 4}(a_j b_j^{\dagger}-a_j^{\dagger}b_j)}S(-r)\ket{0}_{A_j}\otimes S(r)\ket{0}_{B_j}
 \label{eqn:tmssbs}
 \end{equation}
where $a_j$ and $b_j$ are the annihilation operators on the $A_j$ and $B_j$ modes respectively, and $S(r)$ is the single mode squeezing operator~\cite{mandel}. It follows that an $m$-mode TMSS can be prepared across two $m$-mode registers $A$ and $B$ by first negatively squeezing the modes on register $A$ and positively squeezing the modes on register $B$ and then pairwise entangling the modes $A_j$ and $B_j$ (for $j =1$ to $j = m$) using a network of 50:50 beamsplitters. 

As shown in Fig.~\ref{fig:CostFuncs}(d), preparing a TMSS in this manner is the first step of the circuit to measure $C_{\text{\footnotesize LE-TMSS}_{r}}$. The second step is to apply the target unitary $U$ and the inverse of the ansatz $V^\dagger$ to register $A$. The final step is to implement the inverse of the $m$-mode TMSS state preparation in order to measure the overlap with the $m$-mode TMSS. This is done by first inverting the beamsplitter network and then reversing the initial local squeezing. The inverse squeezing can be carried out either by a two-step process consisting of an active optical unitary followed by an on-off photodetection measurement, or in a one-step process by an ideal general-dyne measurement~\cite{yuen}. The probability to obtain the measurement outcome in which all $2m$
modes are in the $\ket{0}$ state, i.e. the vacuum state, is equal to $\vert \bra{\psi^m_{\text{TMSS}}(r)}  UV^\dagger\otimes \id \ket{\psi^m_{\text{TMSS}}(r)} \vert^{2} $. Hence this circuit can be used to measure $C_{\text{\footnotesize LE-TMSS}_{r}}$ as claimed.

\medskip

\paragraph*{Ricocheted Two-Mode-Squeezed State cost}
Unlike the Bell states utilized in discrete variable quantum compiling algorithms, the TMSS only satisfies an approximate ricochet property for finite $r$. That is, with $\ket{\psi_{\text{TMSS}}^{m}(r)}$ defined as in (\ref{eqn:tmssa})
\begin{equation}
   UV^\dagger\otimes \id_{B} \ket{\psi^m_{\text{TMSS}}(r)} \approx U_{A}\otimes V^{*}_{B} \ket{\psi^m_{\text{TMSS}}(r)} 
\end{equation}
with the exact property only holding in the limit that $r \rightarrow \infty$ or for specially chosen $U$ and $V$. 
Consequently, the Ricocheted version of the Two-Mode Squeezed State cost function, i.e. 
\begin{equation}
    C_{\text{\footnotesize R-TMSS}_{r}}(V,U):= 1- \vert \bra{\psi^m_{\text{TMSS}}(r)}  U_{A}\otimes V^{*}_{B} \ket{\psi^m_{\text{TMSS}}(r)} \vert^{2} \, ,
    \label{eqn:costtmss2}
\end{equation}
is  equal to (\ref{eqn:costtmss}) in the limit $r\rightarrow \infty$. The circuit for computing (\ref{eqn:costtmss2}), which is shown in Fig.~\ref{fig:CostFuncs}(e), is identical to the circuit used to measure $C_{\text{\footnotesize LE-TMSS}}$ but with the target and ansatz unitaries prepared in parallel rather than series. The difference between cost functions (\ref{eqn:costtmss}) and (\ref{eqn:costtmss2}) depends on $V$. In Appendix \ref{app:a} we show how the cost functions differ in expectation over finite rank $V$.  The calculation shows that even if the size of $V$ increases multiplicatively, it is sufficient to increase the squeezing parameter $r$ additively in order to make the cost functions (\ref{eqn:costtmss}) and (\ref{eqn:costtmss2}) approximately equal.

Although the $C_{\text{\footnotesize R-TMSS}_{r}}$ cost can be computed by a simple circuit, it has a drawback that its minimum need not be zero when optimizing $V$ for a given $U$. Hence, when used for variational compiling, it will be hard to determine when to terminate the optimization loop. It is therefore helpful to define a normalized version of (\ref{eqn:costtmss2})
\begin{equation}
    \tilde{C}_{\text{\footnotesize R-TMSS}_{r}}(V,U):= 1- \frac{\vert \bra{\psi^m_{\text{TMSS}}(r)}  U_{A}\otimes V^{*}_{B} \ket{\psi^m_{\text{TMSS}}(r)} \vert^{2}}{\mathcal{N}_U \mathcal{N}_V} \, .
    \label{eqn:costtmssnorm}
\end{equation}
where the normalization terms, 
\begin{equation}
    \mathcal{N}_X := \vert \bra{\psi^m_{\text{TMSS}}(r)}  X_{A}\otimes X^{*}_{B} \ket{\psi^m_{\text{TMSS}}(r)} \vert
    \label{eqn:normcostnorm}
\end{equation}
for $X = U$ and $X = V$, (\ref{eqn:normcostnorm}) can be calculated using the same circuit to measure $C_{\text{\footnotesize R-TMSS}_{r}}$. As shown in Appendix~\ref{app:ffn}, this normalized cost $\tilde{C}_{\text{\footnotesize R-TMSS}_{r}}$ is faithful, vanishing if and only if $U$ and $V$ agree up to a global phase. 

Given the need to evaluate the normalization terms, as well as the original cost term, this cost is slightly more resource intensive than the LS-TMSS cost. However, the reduction in circuit depth achieved by using the approximate ricochet property may compensate for this in experimental contexts where coherence lifetimes are short. 

\medskip

\paragraph*{Averaged coherent states cost} It is not possible to define a cost which is directly analogous to Eq.~\eqref{eq:HST_Fbar} in a CV context as there is no direct equivalent to the Haar measure for CV states because of the infinite dimensionality of the Hilbert space for CV systems. Instead, one can consider averaging over a family of states up to a specific energy bound. In Ref~\cite{arrazolaml} a cost is defined in this manner as an average over Fock states. However, large Fock states are hard to produce experimentally, and so we argue that a more natural choice, given the ease with which they can typically be produced in the laboratory, is coherent states. With this in mind, one could consider using the cost function
\begin{align}\label{eq:CScost}
    C_{\text{ACS}_{E}}(V,U)  := 1- \int_{\Vert \bm{\alpha}\Vert^{2}\le E} d\mu(\bm{\alpha})  \big\vert  \langle \bm{\alpha}\vert V^{\dagger}U\vert \bm{\alpha}\rangle \big\vert^{2} 
\end{align}
 where $d\mu(\bm{\alpha})$ is a normalized measure on the set of $m$-mode coherent states $\ket{\bm{\alpha}}$ with energy\footnote{Here, as elsewhere in this paper, we work in units where $\hbar \omega = 1$, where $\omega$ is the mode frequency.} less than $E$. Each of the coherent state overlaps in Eq.~\eqref{eq:CScost} can be computed  using local heterodyne measurements on the $m$ modes. That Eq.~\eqref{eq:CScost} is faithful can be seen from the fact that if it takes the value 0, the modulus of the $Q$-symbol of the unitary operator $V^{\dagger}U$ is equal to 1 almost everywhere on the domain, from which it follows that $V^{\dagger}U=e^{i\phi}$ due to the overcompleteness of coherent states \cite{perelomov}. 
 
 In practice, this cost, which we call the Averaged Coherent State (ACS) cost, can only be estimated by sampling $k$ coherent states with energy less that $E$, i.e. using
 \begin{align}
    C_{\text{ACS}_{E}}(V,U) \approx 1-\frac{1}{k}\sum_{j = 1}^k \big\vert  \langle \bm{\alpha}_{j}\vert V^{\dagger}U\vert \bm{\alpha}_{j}\rangle \big\vert^{2} \, 
    \label{eq:CScostprac}
\end{align}
where $\Vert\bm{\alpha}_j\Vert^2 < E$ for all $j$.
In Section \ref{sec:cvnfl} we will use an NFL theorem for Gaussian operations to argue that $k = 2m$ training states will suffice to learn any Gaussian unitary $U$. In Section~\ref{sec:gng} we provide numerics which suggest that to learn weakly non-Gaussian operations, in particular a small Kerr non-linearity, $k = 2m$ modes is also sufficient. However, in general to learn an arbitrary operation we expect that $k$ will scale with $E$.

Thus, in general, estimating $C_{\text{ACS}_{E}}$ will be more resource intensive than the TMSS costs, $C_{\text{\footnotesize LE-TMSS}_{r}}$ and $ C_{\text{\footnotesize R-TMSS}_{r}}$, in the sense that it requires a larger number of cost evaluations. However, $C_{\text{ACS}_{E}}$ does not require generating large highly-entangled Two-Mode-Squeezed states, and therefore may in some contexts be less experimentally demanding. In particular, we expect this cost to be most useful for learning (approximately) Gaussian operations where the number of training states required is reduced.

\medskip

\subsection{Trainability}\label{sec:Train}

For a variational quantum algorithm to run successfully, i.e. for it to be possible to minimize cost and thereby find the optimum solution, the cost landscape must have sufficiently large gradients to allow for training. Recently, it has been shown that discrete variable VQAs can exhibit so called `barren plateaus', where under certain conditions the gradient of
the cost function vanishes exponentially with the size of the system~\cite{mcclean2018barren,cerezo2021cost,uvarov2020barren,wang2020noise,cerezo2020impact,pesah2020absence,holmes2020barren,arrasmith2020effect,volkoff2021large,marrero2020entanglement,patti2020entanglement,holmes2021connecting,grant2019initialization}. Preliminary results further indicate continuous variable systems~\cite{cvbpl} may exhibit an analogous barren plateau phenomenon where the cost gradients vanish exponentially with the number of system modes. On such barren plateau landscapes (potentially untenably) precise measurements are required
to determine the direction of steepest descent and navigate to the minimum. Thus for any learning algorithm to be \textit{scalable} to large problem sizes it is essential to use a cost that does not exhibit a barren plateau.

Even from basic examples, one can see that the cost functions for CV quantum compiling exhibit a barren plateau. To demonstrate this we will focus on the Loschmidt Echo TMSS cost, but analogous arguments follow for the ricocheted TMSS cost and the averaged coherent state cost. Consider using the Loschmidt Echo TMSS cost to compile the $m$-mode identity operation using the ansatz composed of a product of phase gates, i.e. $V=e^{i\sum_{j=1}^{m}\phi_{j}a_{j}^{\dagger}a_{j}}$ where $\phi_{j}$ are uniform in $[-\pi,\pi]$. Then the cost takes the form
\begin{align}
C_{\text{LE-TMSS}_{r}}(\bm{\phi})&=1 - {1\over \cosh^{4m}r}\prod_{j=1}^{m}\left[ \vphantom{\left(\tanh^{4}r \right)^{-1}}\left( 1-2\cos \phi_{j}\tanh^{2}r \right.\right. \nonumber \\
&{} \left.\left. +\tanh^{4}r \right)^{-1}\right] \, .
\label{eqn:tmssbpl-cost}
\end{align}
It follows (examining $\phi_{1}$ without loss of generality) that 
\begin{align}
E(\vert \del_{\phi_{1}}C_{\text{LE-TMSS}_{r}}\vert)&={1\over (2\pi)^{m}}\int_{0}^{2\pi}d\vec{\phi}\vert \del_{\phi_{1}}C_{\text{LE-TMSS}_{r}}\vert \nonumber\\
&= \left( 2\over \pi(1+2\sinh^{2}r)^{2}\right)^{m}{\tanh^{2}r\over 1+\tanh^{4}r} \, 
\label{eqn:tmssbpl}
\end{align}
which vanishes exponentially with the number of modes $m$.   It therefore follows from Chebyshev’s inequality
\begin{equation}
    P(\vert \del_{\phi}C_{\text{LE-TMSS}_{r}}\vert >\epsilon)\le {E(\vert \del_{\phi}C_{\text{LE-TMSS}_{r}}\vert)\over \epsilon}
    \label{eqn:bplbpltmss}
\end{equation}
that the probability that the cost gradient deviates from zero vanishes exponentially. Thus the landscape exhibits a barren plateau~\cite{cvbpl}. If $r$ is allowed to vary with $m$, then taking sublinear scaling of the total squeezing, e.g., local squeezing $r(m)=O({\ln m^{\alpha}\over m})$, $\alpha >0$, causes (\ref{eqn:bplbpltmss}) to vanish only polynomially.

\begin{figure*}[t!]
    \includegraphics[width=0.95\textwidth]{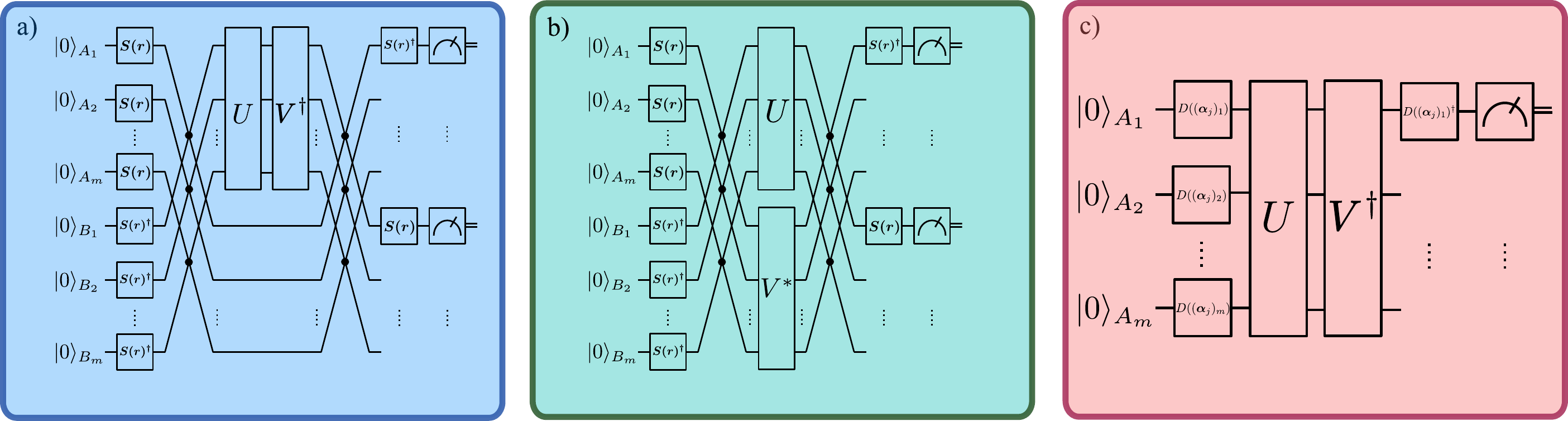}
    \caption{\textbf{Local cost functions.}  Here we show the circuit diagrams for the local versions of the a) Loschmidt Echo TMSS cost $C_{\rm LE-TMSS}^{(L)}$, b) Ricocheted TMSS cost $C_{\rm R-TMSS}^{(L)}$ and c) Averaged Coherent State cost $C_{\rm ACS}^{(L)}$. Crucially, in contrast to their respective global variants, only a pair of modes (in the case of $C_{\rm LE-TMSS}^{(L)}$ and $C_{\rm R-TMSS}^{(L)}$) and a single mode (in the case of $C_{\text{ACS}_{E}}^{(L)}$) is measured per circuit evaluation.} 
    \label{fig:localCostFuncs}
\end{figure*}

For fixed $r$, the barren plateau phenomenon can be circumvented by using a local variant of our proposed costs. Analogously to the local version of the Hilbert-Schmidt test introduced in Ref.~\cite{Khatri2019quantumassisted}, where pairs of qubits, rather than all $2n$ qubits, are measured to compute the local cost; our proposed local TMSS costs can be calculated from measurements on pairs of CV modes instead of $2m$ CV modes. Specifically, as shown in Fig. \ref{fig:localCostFuncs}(a), the local version of the Loschmidt Echo TMSS cost is defined as 
\begin{equation}
C_{\text{LE-TMSS}_{r}}^{(L)}(V,U) := 1 - \frac{1}{m} \sum_{j=1}^m {\rm Pr}(00)_{A_j B_j}
\label{defClocal}
\end{equation}
where $ {\rm Pr}(00)_{A_j B_j}$ is the probability of observing outcome $00$ on the pair of modes $A_j B_j$ from registers $A$ and $B$ in  Fig.~\ref{fig:CostFuncs}(d). This cost function can be shown to be faithful using the same probability theoretic argument used to prove the faithfulness of the local Hilbert Schmidt test cost in Ref.~\cite{Khatri2019quantumassisted}.

The local cost $C_{\text{LE-TMSS}_{r}}^{(L)}$ can be expressed as a sum of entanglement fidelities. To see how, first note that the local marginal states of $UV^{\dagger}\otimes \mathbb{I}_{B}\ket{\psi_{\text{TMSS}}^{m}(r)}$ on $A_{j}B_{j}$, i.e.  
\begin{align}
\rho_{A_j B_j} := \text{Tr}_{\overline{A_{j}}\,\overline{B_{j}}} \left[ (\mathcal{E}_{UV^\dagger} \otimes \mathbb{I}_{B})  \ket{\psi_{\text{TMSS}}^{m}(r)}\bra{\psi_{\text{TMSS}}^{m}(r)} \right] \nonumber  
\end{align}
where $\overline{A_{j}}$ ($\overline{B_{j}}$) is the complement of $A_{j}$ ($B_j$) in mode set $A$ ($B$) and $\mathcal{E}_{UV^\dagger}(...) = U V^\dagger (...) V U^\dagger $, 
can be written as 
\begin{equation}
\rho_{A_j B_j} = \text{Tr}_{\overline{A}_{j}} \left[ \mathcal{E}_{UV^\dagger} \otimes \mathbb{I}_{B_j} \left( \rho^{\text{\tiny TMSS}}_{A_jB_j}(r)\otimes \rho_{\beta(r)}^{\otimes m-1} \right) \right] \, .
\end{equation}
Here $\rho_{\beta(r)}$ is a thermal state at the inverse temperature $\beta(r):=-2\ln \tanh r$,
\begin{equation}
  \rho_{\beta(r)} := {1\over \cosh^{2}r}\ \sum_{n=0}^{\infty}\tanh^{2}r \ket{n}\bra{n}\, ,
\end{equation}
and $\rho^{\text{\tiny TMSS}}_{A_jB_j}(r) := \ket{\psi_{\text{TMSS}}^{m}(r)} \bra{\psi_{\text{TMSS}}^{m}(r)}_{A_jB_j}$.
It follows that $C_{\text{LE-TMSS}}^{(L)}$ can be written as 
\begin{equation}
C_{\text{LE-TMSS}_{r}}^{(L)}(V,U):=1-{1\over m}\sum_{j=1}^{m} F_j \, ,
\label{eqn:tmssloc}
\end{equation}
where $F_j$ is the entanglement fidelity of the channel 
\begin{equation}
\mathcal{E}_{j}(\rho_{A_{j}}):= \text{Tr}_{\overline{A_{j}}} UV^{\dagger}\left(\rho_{A_{j}}\otimes (\rho_{\beta(r)})^{\otimes m-1} \right) VU^{\dagger} \, , 
\end{equation}
with respect to the TMSS. That is,
\begin{equation}
    F_j := \text{Tr} \left[\rho^{\text{\tiny TMSS}}_{A_jB_j} (\mathcal{E}_{j}\otimes 1_{B_{j}}) \left( \rho^{\text{\tiny TMSS}}_{A_jB_j} \right) \right]   \, .
\end{equation}
Thus, not only is the local cost faithful, it also has a natural conceptual interpretation. 

Crucially, the local TMSS cost (\ref{eqn:tmssloc}) appears not to exhibit a barren plateau. For example, for the problem of compiling the identity with multimode phase shifters considered at the beginning of this section, one obtains
 \begin{align}\label{eq:localgrad}
     E(\vert \del_{\phi_{1}}C^{(L)}_{\text{LE-TMSS}_{r}}\vert)&= {2\over \pi m \cosh^{4}r(1+\tanh^{2}r)^{2}}\nonumber \\
     &{} \times \left( m-1+{\tanh^{2}r\over 1+\tanh^{4}r }\right)
 \end{align}
 which, for fixed $r$, is constant as $m\rightarrow \infty$.
 
\medskip
 
We note that, for a fixed number of modes $m$, the costs $C_{\text{LE-TMSS}_{r}}$ and $C_{\text{LE-TMSS}_{r}}^{(L)}$ concentrate to 1 when the squeezing parameter $r$ is large. It follows that the gradients of $C_{\text{LE-TMSS}_{r}}$ and $C_{\text{LE-TMSS}_{r}}^{(L)}$, as seen from Eq.~\eqref{eqn:tmssbpl} and Eq.~\eqref{eq:localgrad}, vanish exponentially with $r$. 
Consequently, training becomes exponentially more resource intensive for larger $r$. A similar exponential vanishing with respect to $r$ was observed for approximations of CV energy-constrained channel fidelities that compare the actions of CV channels on two-mode squeezed states \cite{Sharma2020Opt}. This vanishing gradient problem is conceptually different to the barren plateau phenomenon which may be resolved using a local cost. 
In Section \ref{sec:gng}, we propose a practical resolution for this vanishing gradient problem.

\medskip

 Finally, we note that for the example of compiling the identity operation considered earlier, the averaged coherent state cost function $C_{\text{ACS}_{E}}$ in (\ref{eq:CScost}), and its approximation in (\ref{eq:CScostprac}), do not exhibit barren plateaus if the energy bound $E$ is taken to depend on the mode number $m$ in such a way that the maximal energy per mode  $E(m)/m$ grows sublinearly as a function of $m$ (see Section 2 of Ref.~\cite{cvbpl}). However, we expect that more general compiling problems, such as Gaussian compiling or compiling of Kerr non-linearities, will exhibit barren plateaus. Such trainability issues could again be mitigated by defining a local version of $C_{\text{ACS}_{E}}$. A natural choice in local cost would be (analogously to \eqref{defClocal}) to compute a spatial average of the probability of measuring the vacuum state on each of the modes at end of the circuit in Fig.~\ref{fig:CostFuncs}(f). More concretely, one could use
 \begin{align}
    C_{\text{ACS}_{E}}^{(L)}(V,U)&:= 1- {1\over k}\sum_{j=1}^{k}\Tr[ O_{\rm Local}^{(j)} V^\dagger U\vert \bm{\alpha}_{j}\rangle \langle \bm{\alpha}_{j}\vert U^\dagger V] \\
    &O_{\text{Local}}^{(j)} = {1\over m}\sum_{\ell=1}^{m} |(\bm{\alpha}_{j})_{\ell}\rangle \langle (\bm{\alpha}_{j})_{\ell} \vert_{A_{\ell}} \otimes \mathbb{I}_{\overline{A_{\ell}}} \nonumber \\
    \text{ subject to }&{} \Vert \bm{\alpha}_{j}\Vert^{2} \le E \; \forall \; j \nonumber
\end{align}
where $\ket{\bm{\alpha}_{j}}=\bigotimes_{\ell=1}^{m}\ket{(\bm{\alpha}_{j})_{\ell}}_{A_{\ell}}$ is a coherent state in the $2m$-dimensional phase space, and we have used the discrete version of $C_{\text{ACS}_{E}}(V,U)$ in (\ref{eq:CScostprac}). Computation of one term in the double sum defining $C_{\text{ACS}_{E}}^{(L)}$ is shown in Fig. \ref{fig:localCostFuncs}(c).

\section{Numerical Implementations}\label{sec:gng}

Here we present results for implementing CV quantum compilation to learn commonly encountered CV operations. In particular, we focused on learning arbitrary single mode Gaussian operations, Kerr non-linearities and a general beamsplitter operation.
In each case, we performed continuous parameter optimization in order to minimize the TMSS cost function Eq.~\eqref{eqn:costtmss2}. We focus on the Loschmidt-Echo variant of the cost but similar results are obtained for the Ricocheted variant. We note that it is unnecessary to use the local version of the cost here since for these proof-of-principle implementations we consider learning single and two mode unitaries for which the cost gradients are expected to be manageable even with a global cost.

Given the close connections between the TMSS cost and the HST cost for large $r$, and the operational meaning of $C_{\text{HST}}$ as a measure of the average fidelity between $U$ and $V$, ideally we would use a large $r$ value, i.e. large squeezing, to learn $U$. 
However, as discussed in Section~\ref{sec:Train}, and as demonstrated numerically in Fig.~\ref{fig:Landscape}, the landscape of the TMSS cost becomes overwhelmingly flat for large $r$, making it difficult to train.
We therefore found it more effective to train initially using a small $r$ value. Then once reasonably accurate pre-trained parameters have been obtained using a small $r$, we trained on a larger $r$ to refine the quality of the solution. 

\begin{figure}
    \centering
    \includegraphics[width=0.49\textwidth]{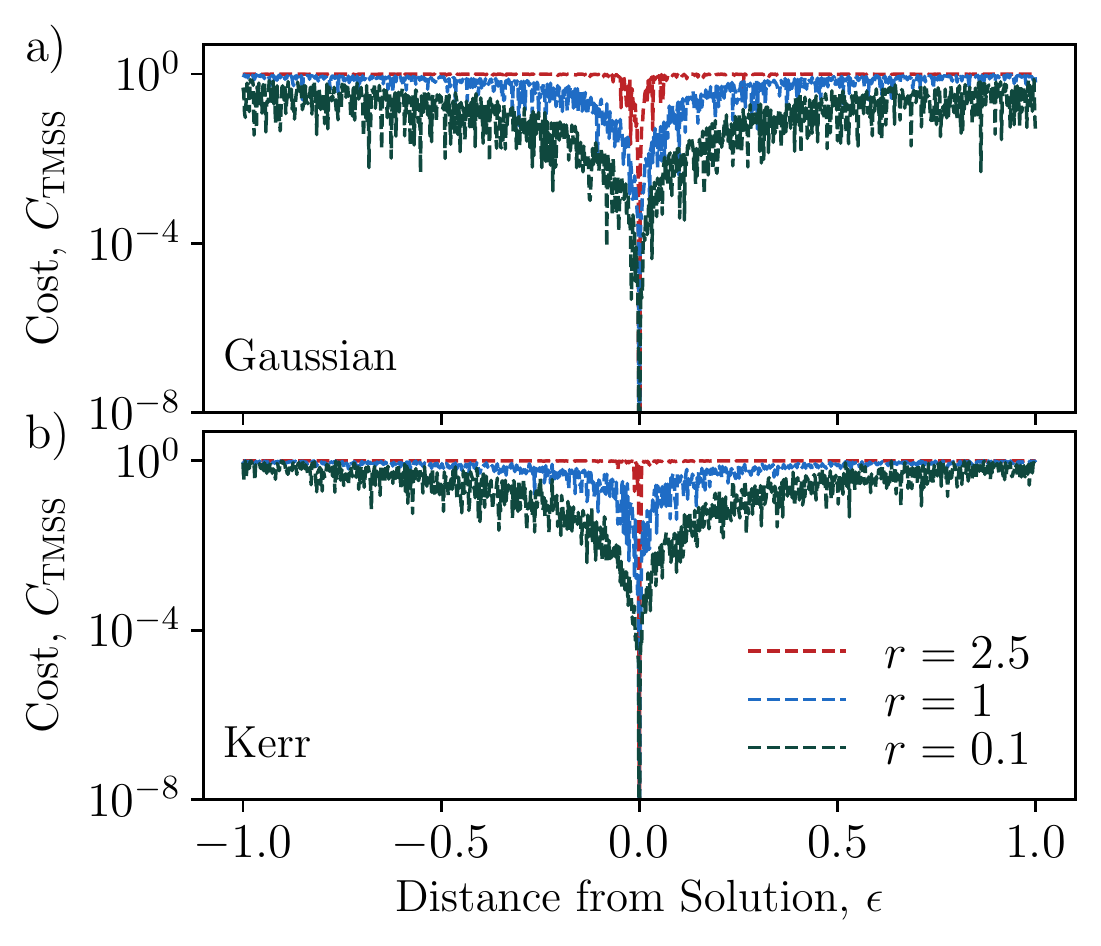}
    \caption{\textbf{Cost Landscapes.} The cost landscape for learning a) an arbitrary Gaussian operation and b) a $\chi = 3$ Kerr non-linearity using 4 layers of the general ansatz defined in Eq.~\eqref{eq:SingleLayerAns}. Here $\epsilon$ is a noise parameter that determines the deviation of the ansatz parameters, $\vec{\theta}$, from the optimum parameters, $\vec{\theta}^{\rm opt}$. Specifically, we set $\theta_k = \theta_k^{\rm opt} + \epsilon R$, where $R$ is a random number between -1 and 1.}
    \label{fig:Landscape}
\end{figure}

\begin{figure*}
    \centering
    \includegraphics[width=0.95\textwidth]{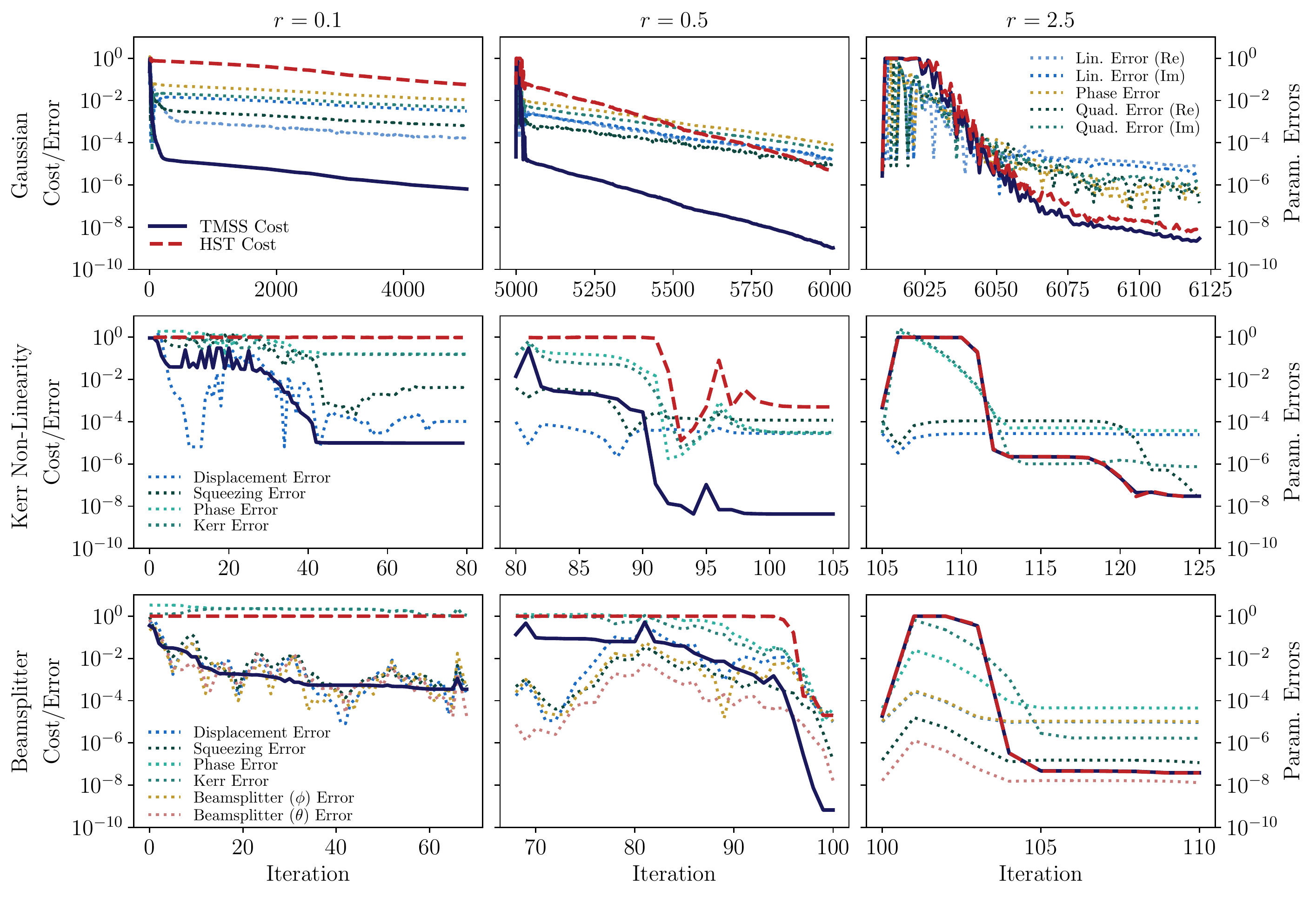}
    \caption{\textbf{Learning using TMSS Cost.} The $C_{{\rm LE-TMSS}_r}$ cost (blue) as a function of iteration for learning a Gaussian (top row), Kerr non-linearity (middle row) and Beamsplitter (bottom row). To mitigate the problem of vanishing cost gradients for large $r$ we take a perturbative approach starting with $r=0.1$ (left column), and after convergence increasing $r$ to $0.5$ (middle column) and then $2.5$ (right column). To quantify the quality of the optimization, we take the optimal parameters obtained at each iteration of the optimisation algorithm and plot both the Hilbert-Schmidt cost $C_{\rm HST}$ (red dashed) and the errors in the individual optimised parameters (dotted). The parameter errors are given in natural units with $\hbar = 1$ and the mass and frequency of the modes equal to 1.}
    \label{fig:TMSSlearning}
\end{figure*}

\medskip
\paragraph*{Gaussian Operations.}
An arbitrary single-mode Gaussian operation
\begin{equation}\label{eq:GausOp}
   U_{\rm Gaus}(\alpha, \beta, \phi) := e^{- i H_{\rm Gaus}(\alpha, \beta, \phi) }
\end{equation}
is generated by the quadratic Hamiltonian
\begin{equation}
    H_{\rm Gaus}(\alpha, \beta, \phi) := \alpha a + \alpha^* a^\dagger + \beta a^2 + \beta^* {a^\dagger}^2 + \phi a^\dagger a
\end{equation}
where $\alpha$ and $\beta$ are arbitrary complex numbers and $\phi$ is an arbitrary real number. We generated a random target Gaussian operation $U_{\rm targ} := U_{\rm Gaus}(\alpha_{\rm targ}, \beta_{\rm targ}, \phi_{\rm targ}) $ by choosing $\Re(\alpha_{\rm targ})$, $\Im(\alpha_{\rm targ})$, $\Re(\beta_{\rm targ})$ and $\Im(\beta_{\rm targ})$ randomly in the range $[0,1]$ and $\phi_{\rm targ}$ in the range $[0, 2\pi]$. We then used $C_{\rm LE-TMSS}$ to learn $U_{\rm targ}$ using an ansatz of the same form. That is, using an ansatz of the form $V_{\rm anz} = U_{\rm Gaus}(\alpha, \beta, \phi)$ where $\alpha$, $\beta$ and $\phi$ are parameters to be variationally learnt. Since $U_{\rm Gaus}$ is readily factorizable into the product of displacement, squeezing and phase operations, this ansatz can be straightforwardly implemented using standard gates on a CV-quantum computer. 

The results of learning\footnote{ To numerically compute the cost functions required for our simulations we worked in a truncated Hilbert space. Specifically, we truncated the Hilbert space to the 50 lowest lying Fock states.} an arbitrary Gaussian operation are shown in the top row of Fig.~\ref{fig:TMSSlearning}. To quantify the quality of the optimization, we take the optimal parameters obtained at each iteration of the optimisation algorithm and plot both the Hilbert-Schmidt cost, $C_{\rm HST}$, and the errors in the individual optimised parameters. As shown in Fig.~\ref{fig:TMSSlearning}(a), we start with $r = 0.1$ and successfully optimize the TMSS cost (using the COBYLA algorithm) down to $10^{-6}$. This corresponds to errors in the HST cost and individual parameters in the region of $10^{-1}$ to $10^{-4}$. We then took the optimal parameters from minimizing the TMSS cost with $r = 0.1$ and optimized using the HST cost with $r = 0.5$. The cost value and parameter errors initially go sharply up (because the old parameters that optimised the cost with $r=0.1$ are no longer optimal) before decreasing again as the new cost is optimised. After optimising with $r= 2.5$ we get both the TMSS and HST costs down to $10^{-9}$ with errors in the individual parameters in the region of $10^{-5}$. Thus Fig.~\ref{fig:TMSSlearning} both demonstrates the effectiveness of our perturbative strategy and highlights how the difference between $C_{\rm TMSS}$ and $C_{\rm HST}$ decreases with increasing $r$. 

\medskip
\paragraph*{Kerr Non-Linearity.} The second optimization task we consider is learning a Kerr non-linearity of the form
\begin{equation}
    U_{\rm Kerr}(\chi) := e^{-i \chi (a^\dagger a)^2} \, .
\end{equation}
Since there is no simple ansatz which can capture an arbitrary non-Gaussian operation, in this case we use the general layered ansatz advocated in Refs~\cite{arrazolaml, killoran2019continuous}. This ansatz is composed of multiple layers that each consist of a displacement, squeeze, phase shift and non-linear Kerr shift. That is, a single layer is of the form
\begin{equation}\label{eq:SingleLayerAns}
    V_{\rm layer}(\alpha, \beta, \phi, \chi) := U_{\rm Kerr}(\chi) R(\phi) D(\alpha) S(\beta) 
\end{equation}
and the total ansatz is composed of a product of $L$ such layers, $V_{\rm ans}(\vec{\theta}) :=  \prod_{l = 1}^{L} V_{\rm layer}(\alpha_l, \beta_l, \phi_l, \chi_l) $ where $\vec{\theta} := \{\alpha_l, \beta_l, \phi_l, \chi_l\}_{l=1}^{L}$. 
Since the gates in every layer constitute a universal set~\cite{killoran2019continuous}, this ansatz can be used to implement any single-mode quantum operation.

We focus on the task of learning a large Kerr non-linearity (of, perhaps, some new, yet to be classified, material). To make this task both non-trivial and physically pertinent we suppose that the Kerr non-linear components used as part of the ansatz are limited to implementing some maximum non-linearity which is less than that of the target non-linearity. Specifically, we suppose that the components of $\vec{\theta}$ are bounded between 0 and 1, and we consider trying to learn  $\chi_{\rm targ} = 3$ using a 4-layered ansatz. 
To perform the optimization we employ the gradient-based Limited-memory BFGS algorithm.

To assess the performance of the optimization in Fig.~\ref{fig:TMSSlearning} we again plot the HST cost as well as a measure of the error in the individual parameters. Given the non-commutativity of the displacement, squeezing, phase shift and Kerr operations, there are multiple possible choices in the parameters $\vec{\theta}$ such that $V(\vec{\theta}) = U_{\rm Kerr}(\chi_{\rm targ})$. Despite this, in practice, we found that the optimization algorithm found the `obvious' solution where the displacement parameters $\alpha_l$, squeezing parameters $\beta_l$, and phase shift parameters $\phi_l$ each sum to zero and the Kerr non-linearity parameters $\chi_l$ summed to $\chi_{\rm targ}$. We therefore took the difference between these values (i.e. $| 
\sum_{l=1} \alpha_l  - 0 |$, $| 
\sum_{l=1} \beta_l  - 0 |$, $| 
\sum_{l=1} \phi_l  - 0 |$ and $| 
\sum_{l=1} \chi_l  - \chi_{\rm targ} |$) as the measure of our displacement, squeezing, phase and Kerr errors respectively.

Similarly to the Gaussian case we find that starting with a small $r$ allows for successful training. Then increasing the value of $r$ improves the quality of the training in the sense that the HST error and parameter errors can be  further decreased. We achieve a final TMSS and HST cost of $10^{-8}$ and parameters errors of $\sim 10^{-5}$. 

\medskip

\paragraph*{Beamsplitter Operation.} Finally, we attempted to learn a beamsplitter operation of the form 
\begin{equation}
    U_{\rm BS}(\theta, \phi) = e^{\theta (a b^\dagger e^{i \phi} - a^\dagger b e^{-i \phi} )} 
\end{equation}
for a two mode system with annihilation operators $a$ and $b$ respectively where $\theta$ and $\phi$ are randomly chosen phases in the range $[0, 2\pi]$. To learn this operation we used a single layer ansatz of the form
\begin{equation}
\begin{aligned}
    V_{\rm layer}^{(12)}(\vec{\theta}) := &U_{\rm BS}(\theta, \phi) V_{\rm layer}^{(1)}(\alpha_1, \beta_1, \phi_1, \chi_1) \\ &\times V_{\rm layer}^{(2)}(\alpha_2, \beta_2, \phi_2, \chi_2) \, 
\end{aligned}
\end{equation}
where $ V_{\rm layer}^{(1)}$ and $ V_{\rm layer}^{(2)}$ indicate the single-mode gate sequence defined in Eq~\eqref{eq:SingleLayerAns} on the first- and second-mode respectively. 
The optimization was successful, with the TMSS and HST costs reduced to $\sim 10^{-7}$. 

We note that while it may superficially appear from Fig.~\ref{fig:TMSSlearning} that fewer iterations are required to learn the Kerr non-linearity and beamsplitter operation than an arbitrary Gaussian operation, this is a feature of our choice in optimisation algorithm. Namely, BFGS uses a gradient based approach which involves evaluating the cost $n_{\rm param}$ times, where $n_{\rm param}$ is the number of parameters that need to be learnt at every iteration step. Once this is accounted it requires more cost evaluations to learn the general beamsplitter or a Kerr non-linearity than to learn a Gaussian operation. This is precisely as one would expect since these are more complex optimization problems.

\section{No-Free-Lunch theorems for CV quantum learning\label{sec:cvnfl}}

In classical machine learning, the No-Free-Lunch (NFL) theorems consider the task of learning a target function $f$, where $f$ maps a discrete input set $X$ to a discrete output set $Y$ (both of size $d$). The learning is performed using a training set $S$ consisting of $|S|$ input-output training pairs, 
\begin{equation}
    S=\{(x_j,f(x_j): \,  x_j\in X  \}_{j=1}^{|S|} \, .
\end{equation}
In the limit of perfect learning, one assumes it is possible to train a hypothesis function $h_S$ to match the target function $f$ on all training pairs in $S$. The No-Free-Lunch theorems then quantify the `generalization error', i.e. how well the hypothesis function matches the target on unseen data. In general terms, the theorems demonstrate that the generalization error of a given learning algorithm is not less than that of a random learning algorithm in expectation over target functions $f$~\cite{wolp,wolp2,wolp3,wolp4, wolfnotess}. That is, the average performance of a learning algorithm is determined not by the choice in learning algorithm but rather by the amount of training data $|S|$. 

Specifically, the generalization error can be quantified by the following \textit{risk} function 
\begin{equation}\label{eq:risk}
    R_f(h_S)=\sum_{x\in X} \pi(x)\mathbf{1}\Big[f(x)\neq h_{S}(x)\Big] ,
\end{equation}
where $\mathbf{1}[S]$ is the indicator function taking value 1 (0) if condition $S$ is satisfied (not satisfied). This is the probability that the hypothesis function $h_{S}(x)$ and target function $f(x)$ differ across $X$, the domain of $f$, when $x$ is sampled from the uniform probability distribution $\pi(x)$. The average risk, averaged over training sets $\SC$ and functions $f$, can for any optimization method be lower bounded as~\cite{wolfnotess}
\begin{equation}\label{eg:classicalNFL}
    \mathbb{E}_f[\mathbb{E}_S[R_f(h_S)]]\geq \left(1-\frac{1}{d}\right)\left(1-\frac{|S|}{d}\right)\,.
\end{equation}
Hence the average risk is determined by the number of training pairs $|S|$, vanishing if and only if $S$ spans the full domain of $f$, i.e. if $|S| = d$.  

\medskip

Similar NFL theorems exist for finite-dimensional quantum circuit learning, in which a target function $f$ corresponds to a unitary quantum channel $U$ and the training set is generalized to a set of quantum state pairs $S=\lbrace \ket{\psi_{j}},U\ket{\psi_{j}}\rbrace_{j=1}^{|S|}$. By defining the generalization error using a suitable distance on quantum state space, it is shown that in general an exponential number of training states, $|S| \sim 2^n$, are required to learn an $n$ qubit unitary~\cite{nopree1}. 

Further, by allowing the training set to consist of pairs of states that are entangled with a reference system, i.e. $S=\lbrace \ket{\psi_{j}},U\otimes\mathbb{I}_{\mathcal{R}}\ket{\psi_{j}} \rbrace_{j=1}^{|S|}$, where $\ket{\psi_{j}}\in \mathcal{H}\otimes \mathcal{H}_{\mathcal{R}}$ are entangled pure states of Schmidt rank $\mathfrak{r}$, an entanglement-enhanced quantum No-Free-Lunch theorem can be derived. In this case, the lower bound of the expected error of a quantum learning algorithm, over all target unitaries $U$ is reduced linearly in $\mathfrak{r}$~\cite{nopree}. This has the important practical implication that, by using entanglement as a resource, the number of unique input-output state pairs needed to learn a target unitary, $U$, may be exponentially reduced  in the limit of perfect learning.

\medskip

In Section~\ref{sec:loug} we derive NFL theorems in a restricted setting where the task is learning a linear optical unitary operation. Specifically, Theorem~\ref{th:ff1} and Theorem~\ref{th:ff3} quantify learning with classical training data (coherent states) and quantum training data (entangled coherent-Fock states) respectively. 
Section \ref{sec:tmssnfl} shows how the entanglement-assisted NFL theorem of \cite{nopree} can be applied in an unrestricted CV learning setting. We further discuss how CV quantum NFL theorems can be used to motivate cost functions for CV quantum compiling. These results are summarized in Table~\ref{tab:NFL-entanglement}.

\subsection{Learning linear optical unitaries from Gaussian training data\label{sec:loug}}

Linear optical unitaries capture the dynamics of multimode beamsplitters, phase shifters and displacement operators. Such unitaries, on $m$ CV modes, can be written in the form $U=e^{iH}$ where $H=H^{\dagger}$ and $[H,\sum_{j=1}^{m}a_{j}^{\dagger}a_{j}]=0$. Here we consider the task of training a hypothesis unitary $V_S$ to emulate a target linear optical unitary $U$ using a set of training data $S$ composed of $m$-mode coherent states. We analyze the expected performance of a generic learning algorithm, over all target linear optical unitaries $U$ and all training sets containing $|S|$ training states.

To fix the notation, an $m$-mode coherent state with mean vector $w$ is written $\ket{w}$. Here $w$ is a row vector in $\mathbb{R}^{2m}$ given by $w=\langle R\rangle$ with $R=(q_{1},p_{1},\ldots,q_{m},p_{m})$ the vector of canonical quadrature operators. 
The action of the target linear optical unitary $U$ on $\ket{w}$ is given by $U\ket{w}=\ket{wO}$ where $O$ is a $2m\times 2m$ orthogonal matrix. The set of $2m$ by $2m$ orthogonal matrices will be denoted $\text{Orth}(2m)$.
Equipped with this notation, the training set to learn a linear optical unitary $U$ using $|S|$ pairs of $m$-mode coherent states can be written as 
\begin{equation}
    S=\lbrace (w_{j},w_{j}O) \rbrace_{j=1}^{\vert S\vert} \in (\mathbb{R}^{2m}\times \mathbb{R}^{2m})^{\times \vert S\vert} \, .
    \label{eqn:coh_trainset}
\end{equation}

Similarly, the action of the hypothesis linear optical unitary $V_S$ on $\ket{w}$ can be written as $V_S\ket{w}=\ket{wT_S}$ where $T_{S}\in\text{Orth}(2m)$.
We focus on the limit of perfect learning and assume that the learning algorithm outputs an orthogonal matrix $T_{S}$ that agrees perfectly with $O$ on all coherent states $w_{j}$ in the training set. That is, we assume that 
$w_{j}T_S=w_{j}O$ for the training data mean vectors $(w_{j}, w_{j} O) \in S$. 

To quantify how well the hypothesis unitary $V_S$ matches the target unitary $U$ on all possible coherent states, i.e. not just the training states, we define a risk function. To do so we utilize a simple loss function of the form  $L(y, z)= \Vert y - z \Vert^{2}$ where $y = x O$ and $z = x T_S$ are the output vectors of the target and hypothesis orthogonal matrices respectively. Throughout this section, $\Vert\cdot\Vert$ refers to the 2-norm on the Euclidean space $\mathbb{R}^{2m}$. The total risk is then defined as the average loss over a multivariate Gaussian distribution of input vectors $x$, i.e. over the distribution
$\pi(x)= {1\over (2\pi\sigma)^{m}}e^{-{\Vert x \Vert^2\over 2\sigma}}$.
The total risk thus takes the form 
\begin{equation}
R_{O}(T_{S})={1\over 8m\sigma}\int d^{2m}x \, \pi(x) L(xT_{S},xO)  \, ,
\label{eqn:rikrik}
\end{equation}
where the normalization factors have been chosen to ensure 
$R_{O}(T_{S})$ takes values between 0 and 1.  
In essence, $R_{O}(T_{S})$ is a measure of how well $T_{S}$ matches $O$ isotropically in phase space. 
Risk values of $R_{O}(T_{S})=0$ and $R_{O}(T_{S})=1$ are both totally informative, corresponding to $T_{S}=O$ and $T_{S}=-O$ respectively, i.e. perfect learning (up to a possible sign error). In contrast, a risk value of $R_{O}(T_{S})=1/2$ implies that the hypothesis unitary matches the target no better than a typical random linear optical unitary. 

The following theorem quantifies the expected risk for learning a linear optical unitary using the training set $S$, \eqref{eqn:coh_trainset}, in the limit of perfect learning. 

\begin{theorem} Let $O$ be distributed according to the normalized Haar measure on $\text{\normalfont{Orth}}(2m)$ and let the training data $S$ of cardinality $\vert S\vert$ be chosen uniformly from a compact connected set of $m$-mode coherent states, as defined in Eq.~\eqref{eqn:coh_trainset}. Then
\begin{equation}
 E_{S}(E_{O}(R_{O}(T_{S})))={1\over 2}-{\vert S\vert \over 4m}
 \label{eqn:th1eq}
 \end{equation}
 \label{th:ff1}
 \end{theorem}
 \begin{proof}
 For fixed $S$, simplify (\ref{eqn:rikrik}) to
 \begin{align}
 R_{O}(T_{S})&={1\over 2} -{1\over 4m\sigma}\int d^{2m}x \, xT_{S}O^{T}x^{T}\pi(x)\nonumber \\
 &= {1\over 2} -{\text{Tr}T_{S}O^{T}\over 4m}
 \label{eqn:rott}
 \end{align}
 
Under the assumption that the learning algorithm outputs $T_{S}$ that agrees with $O$ on the $\vert S\vert$-dimensional subspace of $\mathbb{R}^{2m}$ spanned by the training data (i.e., $w_{j}T=w_{j}O$ for the training data  mean vectors $w_{j}$), we can write 
\begin{equation}
T_{S}O^{T}=\begin{pmatrix}
\mathbb{I}&0\\0&Y
\end{pmatrix}
\label{eqn:yyy}
\end{equation}
where $Y\in \text{Orth}(2m-\vert S\vert)$. Taking the expectation over $O$ gives
\begin{align}
E_{O}(R_{O}(T_{S}))&=\int dO \left[ {1\over 2} -{\text{Tr}T_{S}O^{T}\over 4m} \right]\nonumber \\
&= {1\over 2}-{\vert S\vert \over 4m}-\int dY {\text{Tr}Y\over 4m} \nonumber \\
&= {1\over 2}-{\vert S\vert \over 4m} 
\label{eqn:csnfl}
\end{align}
Because $S$ is taken from a subset of $\mathbb{R}^{2m}$ with no isolated points, one always obtains a set of $\vert S\vert$ linearly-independent coherent states when $S$ is sampled.  Therefore, taking the expectation over $S$ does not change the right-hand side of (\ref{eqn:csnfl}).
 \end{proof}

Theorem~\ref{th:ff1} shows that the generalization error for learning generic linear optic unitaries reduces linearly with the number of pairs of coherent states trained on, vanishing completely for $|S| = 2 m$. 
(We stress that $|S|$ is the number of unique training pairs required to learn the unitary, not the total number, which, due to shot noise and the iterative optimization procedure, will be substantially larger.)
This implies that the Averaged Coherent State cost in Eq.~\eqref{eq:CScostprac} can be approximated using only $2m$ training states when learning $m$-mode linear optical unitaries. More broadly, Theorem~\ref{th:ff1} can be viewed as a ``classical'' NFL theorem for CV systems.

\begin{table*}[ht]
\caption{No-Free-Lunch theorems for CV unitary learning}
\centering
\begin{tabular}{p{0.15\linewidth}p{0.22\linewidth}p{0.2\linewidth}p{0.2\linewidth}p{0.19\linewidth}}
\hline
 Target unitary& Training set $S$ & Entangled training& NFL & Cost Motivated\\
\hline
Linear optical & coherent states & No & ${1\over 2}-{\vert S\vert \over 4m}$ & $C_{\rm ACS_E}$ \eqref{eq:CScostprac} \\
Linear optical & coherent-Fock states & Yes & ${1\over 2}-{\mathfrak{r}\vert S\vert \over 4m}$ & $C_{\rm ECFS}$ \eqref{eq:ECFS} \\
Fock truncated  & Schmidt rank $\mathfrak{r}$ TMSS & Yes & $1-{\mathfrak{r}^{2}\vert S\vert^{2}+d+1 \over d(d+1)}$ & $C_{\rm LE-TMSS}$ \eqref{eqn:costtmss}\\
\hline
\end{tabular}
\label{tab:NFL-entanglement}
\end{table*}
 
\medskip

It is possible formulate a quasiclassical CV NFL theorem in which the training data consists of squeezed, rather than coherent, states. In this case we use a risk function that compares the action of $O$ and $T_{S}$ on the phase space fluctuations of a compact set of centered, pure CV Gaussian states. We find that squeezing in general inhibits the learning process. However, intriguingly, for this definition of the risk, the risk may be reduced by the training set size as a function of $\vert S\vert^{2}$ instead of $\vert S\vert$. This CV NFL theorem is discussed and proved in Appendix \ref{app:bb}.

\medskip

We now show, similarly to the entanglement-assisted discrete variable NFL theorem \cite{nopree}, that utilizing entangled training states can lower the expected risk. This improvement is achieved by modifying the  training data set in Theorem \ref{th:ff1}, while keeping the risk (\ref{eqn:rikrik}) the same. 
Specifically, we now consider a training set 
\begin{equation}\label{eq:SEntangled}
    S = \lbrace (\ket{\psi_{j}^\mathfrak{r}}, U\otimes \mathbb{I}_{\mathcal{R}}\ket{\psi_{j}^\mathfrak{r}} )  \rbrace_{j=1}^{\vert S\vert}\subset (\mathcal{H}_{\mathcal{X}}\otimes \mathcal{H}_{\mathcal{R}})^{\times 2\vert S\vert} \, 
\end{equation}
composed of $|S|$ pairs of $m$-mode \textit{entangled coherent-Fock states} of the form
\begin{equation}
\ket{\psi_{j}^\mathfrak{r}}:={1\over \sqrt{\mathfrak{r}}}\sum_{k=1}^{\mathfrak{r}}\ket{w^{(j)}_{k}}_{\mathcal{X}} \otimes \ket{k}_{\mathcal{R}} \, .
\label{eqn:tdata}
\end{equation}
Here $\{ \ket{w_{k}^{(j)}}_{
\mathcal{X}} \}_{k=1}^{\mathfrak{r}}$ is a set of linearly independent coherent states acting on a system $\mathcal{X}$ and $\ket{k}_\mathcal{R}$ denotes the $k_{\rm th}$ Fock state of an ancilla register $\mathcal{R}$. 
The positive integer $\mathfrak{r}$ acts as an analogue of Schmidt rank in this  context, 
although we note that the linearly independent mean vectors $w_{k}^{(j)}$ need not be approximately orthogonal, so  $\mathfrak{r}$ is not strictly related to the entanglement entropy. To use a precise term, $\mathfrak{r}$ is equal to the exponential of the entropy of coherence \cite{winteryang} with respect to the orthonormal set $\lbrace \ket{w_{k}^{(j)}}\otimes \ket{k} \rbrace_{k=1}^{\mathfrak{r}}$ for any $j$. If $\Vert w_{k}^{(j)}\Vert$ is sufficiently large, $\ket{\psi_{j}}$ has  entanglement entropy approximately equal to $\log_{2} \mathfrak{r}$ with respect to the partition consisting of $m$ CV modes $\mathcal{X}$ and the CV register $\mathcal{R}$   of the training set.

Analogously to the NFL for coherent state training above, we derive the following theorem on the expected risk.

\begin{theorem}
Let $O$ be distributed according to the normalized Haar measure on $\text{\normalfont{Orth}}(2m)$ and let the training data $S$ of cardinality $\vert S\vert$ consist of pairs of entangled coherent-Fock states as defined in Eq.~\eqref{eq:SEntangled} and Eq.~\eqref{eqn:tdata}. Then
\begin{equation}
 E_{S}(E_{O}(R_{O}(T_{S})))={1\over 2}-{\vert S\vert \mathfrak{r} \over 4m}.
 \end{equation}
 \label{th:ff3}
 \end{theorem}
 
 \begin{proof}
 As in the setting of Theorem \ref{th:ff1}, the objective is to learn the orthogonal matrix $O$ corresponding to an $m$-mode linear optical unitary $U$. The assumption of perfect agreement of $T_{S}$ and $O$ on the training data set now corresponds to the condition $V_{\mathcal{S}}\otimes \mathbb{I}_{\mathcal{R}}\ket{\psi_{j}^{\mathfrak{r}}} = U\otimes \mathbb{I}_{\mathcal{R}}\ket{\psi_{j}^{\mathfrak{r}}}$ for all $j$. Proceeding up to (\ref{eqn:rott}) in the same way as in the proof of Theorem \ref{th:ff1}, we now note that the assumption of perfect agreement on training data implies that $w_{k}^{(j)}O = w_{k}^{(j)}T_{S}$ for all $k,j$. Taking into account linear independence of the mean vectors $w_{k}^{(j)}$ in $\mathbb{R}^{2m}$, this means that $T_{S}$ and $O$ are identical on an $\mathfrak{r}\vert S\vert$ dimensional subspace of the phase space $\mathbb{R}^{2m}$. So $T_{S}O^{T} = \mathbb{I}_{\mathfrak{r}\vert S\vert}\oplus Y$ with $Y\in \text{Orth}(2m-\mathfrak{r}\vert S\vert)$ and, instead of (\ref{eqn:csnfl}) above, one gets
\begin{align}
E_{O}(R_{O}(T_{S}))&=\int dO \left[ {1\over 2} -{\text{Tr}T_{S}O^{T}\over 4m} \right]\nonumber \\
&= {1\over 2}-{\mathfrak{r}\vert S\vert \over 4m}-\int dY {\text{Tr}Y\over 4m} \nonumber \\
&= {1\over 2}\left( 1-{\mathfrak{r}\vert S\vert \over 2m} \right).
\end{align}
Again the expectation over training sets $S$ of fixed cardinality $\vert S\vert$ is trivial when the mean vectors $w_{k}^{(j)}$ are chosen uniformly from some compact connected subset of $\mathbb{R}^{2m}$.
 \end{proof}
 
Theorem~\ref{th:ff3} shows that for a fixed training data set size, increasing the parameter $\mathfrak{r}$ in the training data (for large, $\Vert w_{k}^{(j)}\Vert$ this approximately corresponds to increasing the entanglement entropy of the training data) can reduce the generalization error. In this sense, entanglement could be seen to provide a `free-lunch'. However, as with all apparently free lunches, there are caveats. Namely, there may be a hidden cost in obtaining the entangled training data in the first place since entanglement is generally experimentally challenging to create and preserve. Thus how `free' this lunch is will depend on the relative scarcity of training states and entanglement.

It is also important to note that the enhancement provided by entanglement here is less necessary than the enhancement found in the discrete variable case. In the discrete variable case an exponential number of training pairs are required in the absence of entangled training data, whereas to learn linear optical unitaries, the number of unentangled pairs scales linearly in the number of modes.

Theorem~\ref{th:ff3} could be viewed as motivating a cost function of the form 
  \begin{align}
    C_{\text{ECFS}}^{(\mathfrak{r}, k)}(V,U) = 1-\frac{1}{k}\sum_{j = 1}^k \big\vert  \langle \psi_{j}^{\mathfrak{r}}\vert V^{\dagger}U \otimes \mathbb{I}_{\mathcal{R}}\vert \psi_{j}^{\mathfrak{r}}\rangle \big\vert^{2} \, 
    \label{eq:ECFS}
\end{align}
where the $\ket{\psi_{j}^{\mathfrak{r}}}$ are the entangled coherent-Fock states defined in Eq.~\eqref{eqn:tdata}. We note that this is a generalisation of $C_{\rm ACS}$ in the sense that it reduces to $C_{\rm ACS}$ in the limit that $\mathfrak{r} = 1$. To learn a linear optical unitary Theorem~\ref{th:ff3} implies it suffices to use $k = 2m/\mathfrak{r}$ training pairs. One could also potentially use this cost to learn more general unitaries; however, Theorem~\ref{th:ff3} does not apply in that case and therefore one may need to use a significantly larger $\mathfrak{r} k$ to minimise the generalization error.

\begin{figure*}
    \centering
    \includegraphics[width=0.95\textwidth]{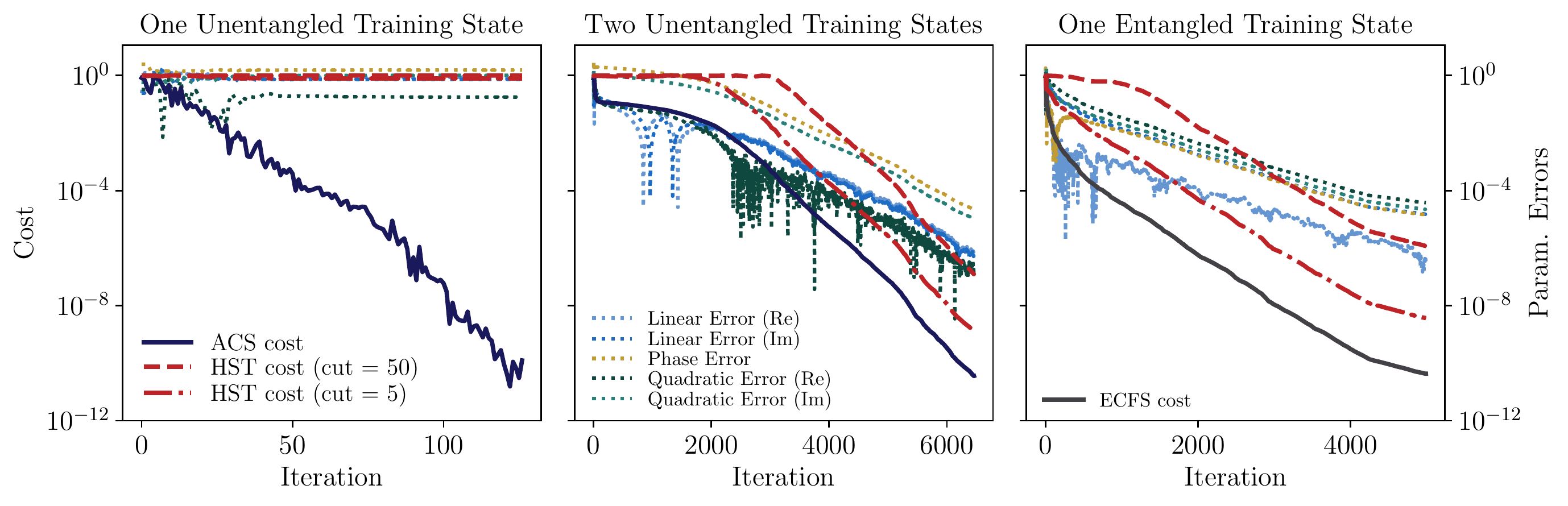}
    \caption{\textbf{Learning Gaussian operations using coherent states and entangled coherent-Fock states.} Cost as a function of iteration for learning a Gaussian operation. In the left and middle plots we use $C_{\rm ACS}$ with a single training pair ($k=1$) and two training pairs ($k=2$) respectively. In the right plot we optimize $C_{\text{ECFS}}^{(\mathfrak{r}=2,k=1)}(V,U)$ which is defined in (\ref{eq:ECFS}) from one training state of the form (\ref{eqn:tdata}). In all cases the coherent states trained on had a random energy up to a maximum of $1$ and a random phase in the range $0$ to $2\pi$. To quantify the quality of the optimization, we take the optimal parameters obtained at each iteration of the optimisation algorithm and plot both the Hilbert Schmidt cost $C_{\rm HST}$ (red) and the errors in the individual optimised parameters in arbitrary units (dotted). The dashed and dot-dashed red lines show the Hilbert-Schmidt cost $C_{\rm HST}$ on the first $50$ Fock states and on the first $5$ Fock states respectively. }
    \label{fig:NFLs}
\end{figure*}

\medskip

In Appendix~\ref{ap:GaussianNFL} we prove that
Theorems~\ref{th:ff1} and \ref{th:ff3} generalize to learning arbitrary Gaussian operations. We thus expect it to be possible to learn a single mode Gaussian operation using a single entangled training pair ($\mathfrak{r}=2$, $|S| = 1$), or two unentangled training pairs ($\mathfrak{r}=1$, $|S| = 2$) but not a single unentangled training pair ($\mathfrak{r}=1$, $|S| = 1$) since for the former the expected risk vanishes whereas the latter corresponds to a finite risk. 

This is indeed supported by our numerical results shown in Fig.~\ref{fig:NFLs} where we optimize $C_{\rm ACS}^{(1)}$ (corresponding to training on a single unentangled training state pair), $C_{\rm ACS}^{(2)}$ (corresponding to training on two unentangled training state pairs) and $C_{\text{ECFS}}^{(2, 1)}$ (corresponding to training on a single entangled training state pair) using the same variational framework set out in Section~\ref{sec:gng}. We find that while it is possible to minimise $C_{\rm ACS}^{(1)}$, this does not correspond to the Gaussian operation being successfully learnt. This is shown by the large learning errors, as measured by the truncated Hilbert-Schmidt Test cost, which quantifies the average error over all possible input states, and individual parameter errors, in the left-hand panel of  Fig.~\ref{fig:NFLs}. Conversely, as shown in the middle- and right-hand panels of Fig.~\ref{fig:NFLs}, when using entangled training data or multiple training states the learning errors are iteratively minimized as the cost is minimized.  

In Fig~\ref{fig:kerr} we present analogous results for the learning of a weak single-mode Kerr non-linearity.
Specifically, as shown in Fig.~\ref{fig:kerr} we find that a single mode ($m=1$) Kerr non-linearity of $\chi = 0.1$ and $\chi = 0.5$ can be learnt using either a single entangled training pair ($\mathfrak{r}=2$, $|S| = 1$), or two unentangled training pairs ($\mathfrak{r}=1$, $|S| = 2$) but not a single unentangled training pair ($\mathfrak{r}=1$, $|S| = 1$).

\begin{figure*}
    \centering
    \includegraphics[width=0.95\textwidth]{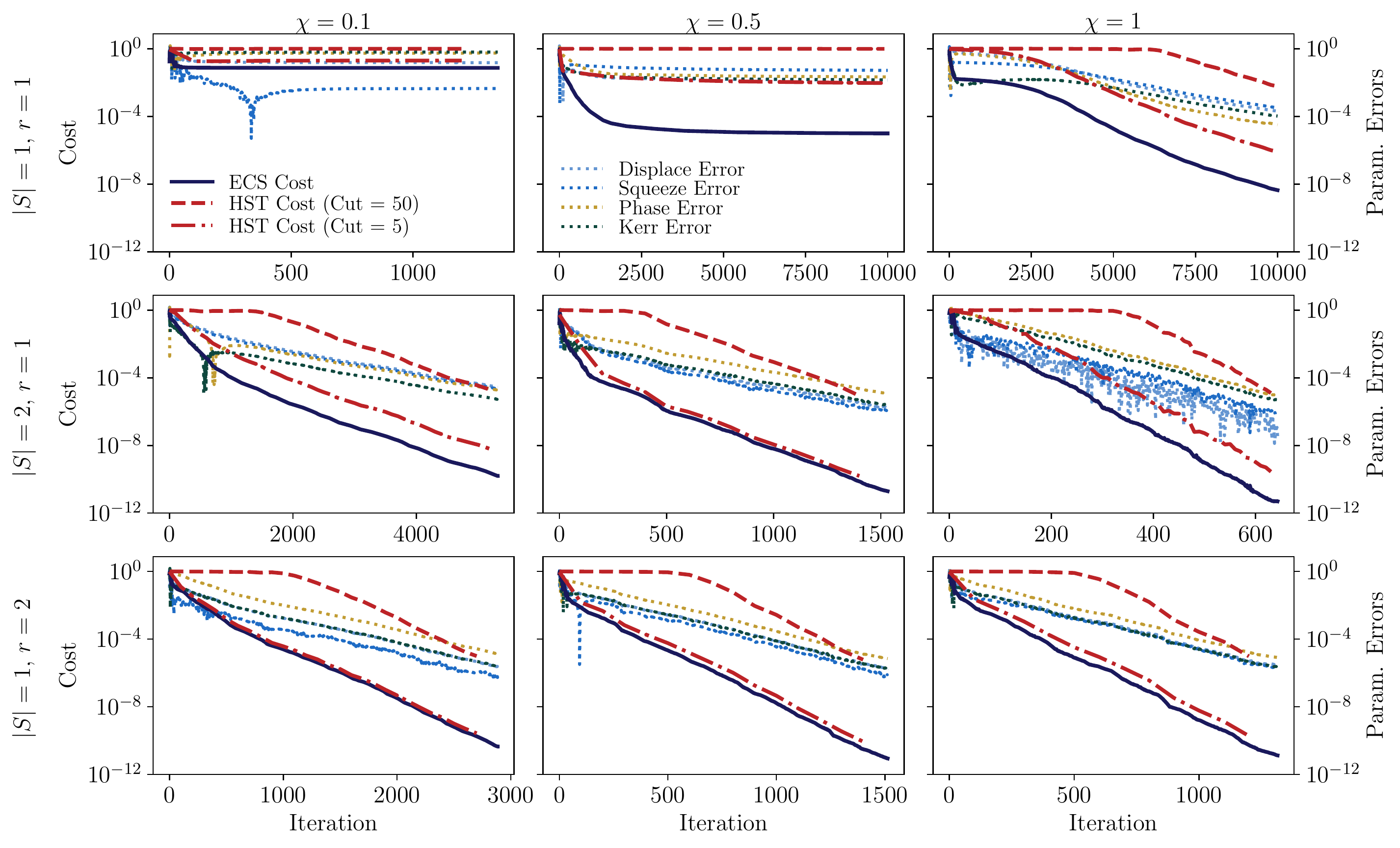}
    \caption{\textbf{Learning Kerr non-linearities using coherent states and entangled coherent-Fock states.} Cost as a function of iteration for learning a $\chi = 0.1$ (top) and $\chi = 0.5$ (bottom) non-linearity. In the left and middle columns we use $C_{\rm ACS}$ with a single training pair ($k=1$) and two training pairs ($k=2$) respectively. In the right hand column we use $C_{\rm ECFS}$ with a single entangled training pair ($k=1$ and $\mathfrak{r} = 2$). In all cases the basic coherent state training states have a random energy up to a maximum of $1$ and a random phase in the range $0$ to $2\pi$. To quantify the quality of the optimization, we take the optimal parameters obtained at each iteration of the optimization algorithm and plot both the Hilbert-Schmidt cost $C_{\rm HST}$ (red) and the errors in the individual optimized parameters in arbitrary units (dotted). The dashed red line and dot-dashed red lines show the truncated Hilbert-Schmidt cost $C_{\rm HST}$ on the first $50$ Fock states and on the first $5$ Fock states respectively.}
    \label{fig:kerr}
\end{figure*}

\subsection{Learning arbitrary unitaries and motivation of compiling cost functions\label{sec:tmssnfl}}

Theorems~\ref{th:ff1} and \ref{th:ff3} concern learning linear optical unitaries. The question remains whether similar entanglement assisted NFL theorems can be derived for learning arbitrary CV unitaries. 

To answer this question, it is useful to recall that the discrete variable (i.e., finite dimensional) entanglement-assisted quantum NFL theorem in Ref.~\cite{nopree}. Specifically, the analog of Theorem \ref{th:ff3} takes the form
\begin{equation}
    E_{S}(E_{U}(R_{U}(V_{S}))) = 1-{\mathfrak{r}^{2}\vert S\vert^{2}+d+1 \over d(d+1)}
    \label{eqn:er}
\end{equation}
where $U$ is the target unitary, $V_{S}$ is the output of the learning algorithm on entangled training states in $S$, $\mathfrak{r}$ is the Schmidt rank of the training data states, and the risk is
\begin{equation}
    R_{U}(V_{S}):= {1\over 4}\int d\psi \Vert U\ket{\psi}\bra{\psi}U^{\dagger}-V_{S}\ket{\psi}\bra{\psi}V_{S}^{\dagger}\Vert_{1}^{2}.
    \label{eqn:er2}
\end{equation}
In (\ref{eqn:er2}), the integral is over all pure states according to the Haar measure.  

A continuous variable NFL theorem cannot be derived that is strictly analogous to \eqref{eqn:er} in the discrete variable setting because there is no Haar measure over the unitary group in $B(\mathcal{H})$ for infinite dimensional $\mathcal{H}$. On the other hand, one is often only interested in the action of the target unitary $U$ on Fock states only up to a finite cutoff. For example, recent proposals for efficient updates and derivatives of Gaussian gates in parameterized CV circuits utilize cutoff recursion relations for the Fock matrix elements of the gates \cite{Miatto2020fastoptimizationof}.

Eq.~\eqref{eqn:er} implies that a single full rank state, i.e. a state with rank $d$, can be used to fully learn a unitary of rank $d$. Thus, the truncated TMSS states defined in Eq.~\eqref{eqn:tmsslim} can be used to learn arbitrary $d$ dimensional unitaries without incurring a generalization error. Taking the limit that $d$ tends to infinity, this implies that Loschmidt-Echo TMSS cost can be used to learn arbitrary CV unitaries, thus further motivating its use.

\section{Discussion}\label{sec:disc}

In this work we have established a framework for quantum compiling in continuous variable systems. We started by motivating the TMSS cost (both the Loschmidt Echo and Ricocheted variants) and the averaged coherent state cost as natural CV analogues of the Hilbert-Schmidt state cost. Our numerical implementations demonstrated the successful learning of single mode Gaussian operations, a generalized Beamsplitter operation and Kerr non-linearities using these costs.

We subsequently showed how these costs may be alternatively motivated via a series of increasingly general `(No-)Free Lunch' theorems. Firstly, the NFL theorem for Gaussian operations using coherent state mean vector training data establishes that it is possible to perfectly learn an $m$-mode Gaussian operation by training on only $2m$ coherent states. This implies that it is possible to learn arbitrary Gaussians by training on an approximation of the averaged coherent state cost using only $2m$ coherent states. Next, the NFL theorem for Gaussian operations using entangled coherent-Fock states both showed how entanglement may be used to reduce the amount of training data required to learn Gaussian operations and motivated an alternative entanglement-enhanced cost function for compiling that makes use of entangled coherent-Fock states. Finally, we argued that taking the continuum limit of the discrete variable entanglement-enhanced NFL theorem implies that to learn an arbitrary unitary on a single training state requires a full rank state. This motivates training using the TMSS cost. 

It is worth highlighting that these (No-)Free-Lunch theorems quantify the number of \textit{different} training pairs required to learn a unitary in the ideal case of perfect training.  That is, they \textit{do not} give the \textit{total} number of copies of training pairs that are required to learn the unitary. Indeed, given shot noise, a large number of copies of each pair will in fact be required to evaluate the cost.  More generally, training may be imperfect not only due to shot noise but also hardware noise or the presence of barren plateaus or local minima in the training cost function landscape. A valuable extension would be to generalize the theorems to account for imperfect learning. 

It would also be interesting to derive further NFL theorems for alternative classes in training data. For example, one might be concerned with learning a unitary $U$ from homodyne or heterodyne detection data, in which case a risk function could be defined in terms of the difference in the expected quadrature vector of the output state for the hypothesis and target unitaries. 
General unitary learning protocols based on other CV measurement-motivated risk functions, such as those associated with CV distinguishability norms~\cite{bdh2,bdh}, are expected to have associated NFL theorems and quantum compiling protocols that are adapted to the measurement class under consideration.

We further note that Two-Mode-Squeezed states are not the only choice of state to saturate the entanglement enhanced NFL bound for arbitrary unitaries. One could alternatively use any full rank state, such as cluster states. A finite energy CV cluster state is defined by $\ket{\text{CL}_{r}}=e^{iq\otimes q}(S(-r)\ket{0})^{\otimes 2}$, where $q$ is the single-mode position quadrature and $S(r)$ is the unitary squeezing operator. The state $\ket{\text{CL}_{r}}$ limits to the well-known CV cluster state for $r\rightarrow \infty$~\cite{PhysRevLett.97.110501}. One could use $\ket{\text{CL}_{r}}$ to define a faithful cost function analogous to the Loschmidt Echo and Ricocheted TMSS costs. 

As quantum hardware develops, the CV quantum compiling algorithms we have presented here are expected to find use optimizing short depth CV quantum circuits, thereby aiding the implementation of larger scale quantum algorithms. Further, we envision that tuning CV quantum resources such as intensity or squeezing could allow one to implement our CV quantum compiling algorithms in a noise resistant way. For example, results of Ref.~\cite{cvbpl} indicate that sublinear scaling (with mode number $m$) of coherent state intensity and number of quantum-limited attenuator layers does not induce barren plateaus in cost functions such as $C_{\text{ACS}_{E}}$ when restricted to linear optical unitaries. More generally, we are excited by the idea that these quantum compilation algorithms may be used to study the optical properties of new materials. It would be interesting to explore whether these algorithms could be combined with meta-learning strategies to actively design new materials with desirable properties such as controllable squeezing amplitudes or non-linearities.

\medskip

\acknowledgements

The authors thank Kunal Sharma and Patrick Coles for helpful discussions. TV acknowledges support from the LDRD program at LANL. ZH acknowledges support and AS acknowledges initial support from the LANL ASC Beyond Moore's Law project. This material is based upon work supported by the U.S. Department of Energy, Office of Science, National Quantum Information Science Research Centers, Quantum Science Center (AS).
\bibliography{phasebib.bib}

\onecolumngrid
\appendix
\section{\label{app:a}Comparison of Loschmidt-Echo (\ref{eqn:costtmss}) and Ricocheted (\ref{eqn:costtmss2}) TMSS costs}

We show how the difference between cost functions (\ref{eqn:costtmss}) and (\ref{eqn:costtmss2}) depends on the squeezing parameter $r$ and the rank of the variational ansatz $V$. Consider the truncated two-mode squeezed state in (\ref{eqn:tmsslim}) and variational {\it Ansatz} $V$ such that the rank of $V$ is $\mathfrak{r}$. For $V$ distributed with respect to the Haar measure on the unitary group $U(\mathfrak{r})$, one can see that the states $V^{\dagger}\otimes \mathbb{I}\ket{\psi^{\mathfrak{r}}_{\text{TMSS}}(r)}$ and $\mathbb{I}\otimes V^{*}\ket{\psi^{\mathfrak{r}}_{\text{TMSS}}(r)}$ are nearly equal in expectation for large $r$. Specifically, the expected modulus of the inner product of these states is given by
\begin{align}
    E_{V}\left( \Big\vert \left( V^{\dagger}\otimes \mathbb{I}\ket{\psi^{\mathfrak{r}}_{\text{TMSS}}(r)},\mathbb{I}\otimes V^{*}\ket{\psi^{\mathfrak{r}}_{\text{TMSS}}(r)} \right)\Big\vert \right)&= {1-\tanh^{2}r\over 1-\tanh^{2\mathfrak{r}}r}E_{V}\left( \sum_{\ell,\ell ' =0}^{\mathfrak{r}}\tanh^{\ell+\ell'}r\vert V_{\ell,\ell'}\vert^{2} \right)\nonumber \\
    &= {1\over \mathfrak{r}}{1-\tanh^{2}r\over 1-\tanh^{2\mathfrak{r}}r} \left( {1-\tanh^{\mathfrak{r}}r \over 1-\tanh r } \right)^{2} \nonumber \\
    &={1\over \mathfrak{r}}{1+\tanh r\over 1+\tanh^{\mathfrak{r}}r}{1-\tanh^{\mathfrak{r}}r\over 1-\tanh r}\nonumber \\
    &\sim \tanh^{\mathfrak{r-1}}r \; \text{ for } r\rightarrow \infty
    \label{eqn:aooo}
\end{align}
For fixed $\mathfrak{r}$, the $r\rightarrow \infty$ limit is 1. Further, even if $\mathfrak{r}$ is increased by a multiplicative factor $\mathfrak{r}\mapsto \lambda \mathfrak{r}$, i.e., the unitaries considered are in $U(\lambda \mathfrak{r})$, the value of the expectation remains close to 1 if one simply adjusts the squeezing $r$ according to $r\mapsto r + \ln \lambda$. To see this, just expand the asymptotic function in (\ref{eqn:aooo}) with respect to the small number $e^{-2r}$ to get $\tanh^{\mathfrak{r}-1}r \sim 1-2(\mathfrak{r}-1)e^{-2r}$ for large $r$.

\section{\label{app:ffn}Faithfulness of $\tilde{C}_{\text{R-TMSS}_{r}}$}

To prove the faithfulness of $\tilde{C}_{\text{R-TMSS}_{r}}$ we start by showing that it can be written in terms of the inner product $(V,U)_{\rho_{\beta}^{\otimes m}}:=\text{Tr}\sqrt{\rho_{\beta}}^{\otimes m}U\sqrt{\rho_{\beta}}^{\otimes m}V^{\dagger}$, where $\rho_{\beta}\propto \sum_{n=0}^{\infty}e^{-\beta n}\ket{n}\bra{n}$ is a single-mode thermal state with inverse temperature $\beta$. This inner product appears in the theory of generalized conditional expectations \cite{petz} and quantum relative entropies \cite{ruskai}. Despite its complicated appearance, $(V,U)_{\rho_{\beta}^{\otimes m}}$ is actually efficiently computable using pure entangled Gaussian state preparation. To demonstate this fact, consider $m$ copies of a two-mode squeezed state prepared in mode pairs $(A_{j},B_{j})$, $j=1,\ldots, m$:
\begin{equation}
    \ket{\psi^m_{\text{TMSS}}(r)} \propto \bigotimes_{j=1}^{m}\sum_{\ell=0}^{\infty}(\tanh r)^{\ell}\ket{\ell}_{A_{j}} \otimes \ket{\ell}_{B_{j}}
\end{equation} with squeezing parameter $r$ satisfying $-2\ln \tanh r = \beta$. It follows that \begin{align}
(V,U)_{\rho_{\beta}^{\otimes m}}&= \left( {1-e^{-{\beta }}}\right)^{m} \sum_{\vec{\ell},\vec{\ell}'} e^{-{\beta \over 2}\Vert \vec{\ell}+\vec{\ell}'\Vert_{1}}U_{\vec{\ell},\vec{\ell}'}(V^{\dagger})_{\vec{\ell}',\vec{\ell }}\nonumber \\
&= \left( {1-\tanh^{2}r}\right)^{m}\sum_{\vec{\ell},\vec{\ell}'} (\tanh r)^{\Vert \vec{\ell} +\vec{\ell}'\Vert_{1}} U_{\vec{\ell},\vec{\ell}'}{V}_{\vec{\ell },\vec{\ell}'}^{*}\nonumber \\
&= \text{Tr}\ket{\psi^m_{\text{TMSS}}(r)}\bra{\psi^m_{\text{TMSS}}(r)}^{\otimes m} U_{A}\otimes V_{B}^{*}
\label{eqn:thtmss}
\end{align} where the sums over $\vec{\ell},\vec{\ell}'$ are over $\mathbb{Z}_{\ge 0}^{\times m}$. It follows from the definition (\ref{eqn:costtmss2}) that $C_{\text{R-TMSS}_{r}}(V,U)=1-\vert (V,U)_{\rho_{\beta}^{\otimes m}} \vert^{2}$ and from (\ref{eqn:costtmssnorm}) that $\tilde{C}_{\text{R-TMSS}_{r}}(V,U)=1-{\vert (V,U)_{\rho_{\beta}^{\otimes m}} \vert^{2}\over \vert (V,V)_{\rho_{\beta}^{\otimes m}} \vert^{2}\vert (U,U)_{\rho_{\beta}^{\otimes m}}\vert^{2}}  $. The fact that $(X,Y)_{\rho_{\beta}^{\otimes m}}$ is linear in $Y$, conjugate linear in $X$, $(X,X)_{\rho_{\beta}^{\otimes m}} \in \mathbb{R}$ (with value 0 if and only if $X=0$) are clear.  Therefore, the Cauchy-Schwarz inequality
\begin{equation}
    \vert (V,U)_{\rho_{\beta}^{\otimes m}} \vert^{2} \le \vert (V,V)_{\rho_{\beta}^{\otimes m}}\vert\vert (U,U)_{\rho_{\beta}^{\otimes m}} \vert
\end{equation}
holds. It implies the faithfulness of (\ref{eqn:costtmssnorm}), i.e.,  $\tilde{C}_{\text{R-TMSS}_{r}}(V,U)=0$ if and only if $U=e^{i\phi}V$ for some $\phi \in [0,2\pi)$.

\section{CV NFL theorems for Gaussian operations}\label{ap:GaussianNFL}

Here, we show that Theorem~\ref{th:ff1} and  Theorem~\ref{th:ff3} can be generalized to learning arbitrary Gaussian operations. 

\smallskip

In Corollary \ref{cor:opop} below, the target unitary $U$ is associated with $L\in \text{Sp}(2m,\mathbb{R})$ by $U^{\dagger}RU=RL$ and the learning algorithm outputs $T_{S}\in \text{Sp}(2m,\mathbb{R})$ when given training set $S$ in (\ref{eqn:coh_trainset}). A $2m\times 2m$ real matrix $L$ is symplectic iff $L^{T}\Delta L=\Delta$, where $\Delta=\begin{pmatrix}0&1\\-1&0\end{pmatrix}^{\oplus m}$ is the standard symplectic form on $\mathbb{R}^{2m}$.

\begin{corollary}
 For any $m\in \mathbb{N}$, define   $\mathcal{G}^{(2m)}:=\lbrace O_{1}ZO_{2}\rbrace$ where $O_{1},O_{2}\in \text{\normalfont{Orth}}(2m)\cap \text{\normalfont{Sp}}(2m,\mathbb{R})$ and $Z=\bigoplus_{j=1}^{m}\text{\normalfont{diag}}(z_{j},z_{j}^{-1})$ where $z_{j}\in \mathbb{R}_{+}$. Let $S$ be the training data (\ref{eqn:coh_trainset}) with $\vert S\vert \equiv 0\mod 2$, and let $L\in \textup{Sp}(2m,\mathbb{R})$. If $R_{L}(T_{S})$ is the risk (\ref{eqn:rikrik}) and $\mathcal{G}^{(2m)}$ is equipped with the probability measure $dO_{1}\;dO_{2}\;\mu(d\vec{z})$, with $dO_{1}$ and $dO_{2}$ Haar measure and $\mu(d\vec{z})$ a probability measure on $\mathbb{R}^{m}$, then  $E_{S}(E_{L}(R_{L}(T_{S})))$ is given by (\ref{eqn:th1eq}).
 \label{cor:opop}
\end{corollary}
\begin{proof}
Every symplectic matrix $L$ is in $\mathcal{G}^{(2m)}$
due to the Bloch-Messiah decomposition \cite{serafini}. From (\ref{eqn:rott}), it follows that $R_{L}(T_{S})={1\over 2}-{\text{Tr}T_{S}L^{T}\over 4m}$. Since $T_{S}$ and $L$ are assumed to agree on the subspace of $\mathbb{R}^{2m}$ spanned by the training data, one can write $T_{S}L^{T}=\begin{pmatrix}\mathbb{I}_{\vert S\vert} & A \\ B &Y\end{pmatrix}$. But $T_{S}L^{T}\in\text{Sp}(2m,\mathbb{R})$ implies that
\begin{align}
    \Delta_{\vert S\vert}+B^{T}\Delta_{2m-\vert S\vert}B&=\Delta_{\vert S\vert} \nonumber \\
    Y^{T}\Delta_{2m-\vert S\vert}Y+A^{T}\Delta_{\vert S\vert}A&=\Delta_{2m-\vert S\vert}
    \label{eqn:sympcond}
\end{align}
where $\Delta=\Delta_{\vert S\vert}\oplus \Delta_{2m-\vert S\vert}$. The first equation in (\ref{eqn:sympcond}) implies $B=0$ from non-degeneracy of the symplectic form. Let $V$ be the Gaussian unitary that acts on the canonical operators as $V^{\dagger}RV=RT_{S}L^{T}$. Then the action of $V$ on the $2m$-vector of operators $(0,\ldots, 0,\tilde{R})$ with $\tilde{R}:=(q_{\vert S\vert+1},p_{\vert S\vert+1},\ldots, q_{m},p_{m})$ is $(\tilde{R}B, \tilde{R}Y) = (0,\ldots,0,\tilde{R}Y)$, so the unitary invariance of the canonical commutation relation implies that $Y\in \text{Sp}(2m-\vert S\vert)$. It then follows from the second equation of (\ref{eqn:sympcond}) that $A=0$.  One concludes that
\begin{align}
    E_{L}(R_{L}(T_{S}))&={1\over 2} - {\vert S\vert \over 4m} - E_{Y}\left( {\text{Tr}Y\over 4m} \right)
    \label{eqn:drdrdr}
\end{align}
where $E_{Y}$ is taken with respect to the measure on $Y$ induced by the measure on $\mathcal{G}^{(2m)}$. Since $Y$ is a symplectic matrix, it can be written $W_{1}FW_{2}$ with $W_{1},W_{2}\in \text{Orth}(2m-\vert S\vert)$. The fact that $W_{1}$ and $W_{2}$ are independent and distributed according to Haar measure follows from the restricting the Haar measure on $O_{1}$ and $O_{2}$. Therefore, the expectation over $Y$ in (\ref{eqn:drdrdr}) is zero. 
\end{proof}

An entirely equivalent argument can be used to generalize Theorem~\ref{th:ff3} to learning Gaussian operations.

\section{\label{app:bb}CV NFL theorem with Gaussian training data}
A centered Gaussian state is a Gaussian state that satisfies $\langle R \rangle =0$ and, therefore, is uniquely defined by its covariance matrix $\Sigma_{i,j}={1\over 2}\langle [R_{i},R_{j}]_{+}\rangle$, which is a positive $2m\times 2m$ matrix. For example, the only centered coherent state is the vacuum $\Sigma_{\ket{0}}=\text{diag}({1\over 2},\ldots,{1\over 2})$. For examples with entanglement, the two-mode squeezed states and CV cluster states are pure, centered Gaussian states with $m=2$. Let $\ket{\psi}$ be a centered Gaussian state and let $U$ be a linear optical unitary that satisfies $U^{\dagger}RU = RO$. Then $\Sigma_{U\ket{\psi}}=O^{T}\Sigma_{\ket{\psi}}O$. Instead of training with coherent state mean vectors, consider now linearly independent training data  $S=\lbrace (\Sigma^{(j)},O^{T}\Sigma^{(j)}O) \rbrace_{j=1}^{\vert S\vert}$ where $\Sigma^{(j)}$ is the covariance matrix of an $m$-mode pure, centered Gaussian state satisfying $\text{rank}{\mathbb{I}_{2m}\over 2} - \Sigma^{(j)}=2$ for each $j$, i.e., the state is squeezed only in one phase space direction.  We consider the risk function
\begin{align}
R_{O}(T_{S})&={1\over m(\log D)^{m}}\int_{\Omega}d\Sigma\int  \Vert T^{T}_{S}\Sigma T_{S} - O^{T}\Sigma O\Vert_{2}^{2}
\label{eqn:risk2}
\end{align}
where the integral $d\Sigma$ is taken over a compact subset $\Omega$ of covariance matrices that satisfy $\Vert \Sigma \Vert \le {D\over 2}$ and $\vert \text{det} 2\Sigma \vert = 1$. Physically, $\Omega$ is the set of covariance matrices of pure Gaussian states with maximal squeezing parameter $r= {1\over 2}\log D$. We assume that the learning algorithm outputs an orthogonal matrix $T_{S}$ such that $T^{T}_{S}\Sigma^{(j)}T_{S}=O^{T}\Sigma^{(j)} O$ for all $j$, i.e., the algorithm produces perfect agreement with the target on the training data set. With the cost function (\ref{eqn:risk2}), Theorem \ref{th:ff} shows that the expected risk is reduced by a function scaling as $\vert S\vert^{2}$ instead of $\vert S\vert$ in Theorem \ref{th:ff1}. 

\begin{theorem}Let $O$ be distributed according to the normalized Haar measure on $\text{\normalfont{Orth}}(2m)$ and let the training data $S$ of cardinality $\vert S\vert$ be chosen uniformly from a compact subset of $m$-mode pure, centered Gaussian states satisfying the rank condition above. Then for the risk function (\ref{eqn:risk2}),
\begin{equation}
E_{S}(E_{O}(R_{O}(T_{S})))= {D^{2}-D^{-2}\over 4 \log D}\left(1-{1\over2(m+1)}\right)  - { (D-D^{-1})^{2}(\vert S\vert^{2}+1) \over 8m(\log D)^{2}} + \mathcal{O}(m^{-2}).
\label{eqn:r2}
\end{equation}
\label{th:ff}
\end{theorem}

\begin{proof}
The integral defining the risk (\ref{eqn:risk2}) is over a compact set $\Omega$ of pure, centered Gaussian states that have squeezing parameters with magnitude uniformly distributed between $r=-{1\over 2}\log D$ and $r={1\over 2}\log D$ where $D>1$. Specifically, the covariance matrices $\Sigma$ appearing in the integral have the form 
\begin{equation}\Sigma \in \left\lbrace W^{T}\text{diag}\left({e^{-2r_{1}}\over 2},{e^{2r_{1}}\over 2},\ldots ,{e^{-2r_{m}}\over 2},{e^{2r_{m}}\over 2}\right)W : W\in \text{Orth}(2m) \, , \, r_{j}\in \left[-{1\over 2}\log D,{1\over 2}\log D\right] \right\rbrace.
\end{equation}
For calculating the expected risk function, it is  advantageous to use the ${\bm{vec}}$ functor. For a matrix $A\in \text{End}(\mathbb{R}^{2m})$,
\begin{equation}
{\bm{vec}}A = \sum_{i=1}^{2m}\sum_{j=1}^{2m}A_{i,j}e_{i}\otimes e_{j} \in \mathbb{R}^{4m^{2}}
\end{equation}
where $\lbrace e_{j}\rbrace_{j=1}^{2m}$ is an orthonormal basis of $\mathbb{R}^{2m}$. For example ${\bm{vec}}(W^{T}A W) = (W\otimes W) {\bm{vec}}(A)$. Also, ${\bm{vec}}$ is an isometry from $\text{End}(\mathbb{R}^{2m})$ as a finite-dimensional Hilbert space with Hilbert-Schmidt inner product to $\mathbb{R}^{4m^{2}}$ as a Hilbert space with Euclidean inner product: $\Vert A \Vert_{2}^{2}=\Vert {\bm{vec}}A \Vert^{2}$. The risk (\ref{eqn:risk2}) becomes
\begin{footnotesize}
\begin{align}
R_{O}(T_{S})&={1\over m(\log D)^{m}}\int_{-{1\over 2}\log D}^{{1\over 2}\log D}dr_{1}\cdots \int_{-{1\over 2}\log D}^{{1\over 2}\log D}dr_{m} \int dW \nonumber \\
&{} \left[  \Vert T_{S}^{T} W^{T}\text{diag}\left({e^{-2r_{1}}\over 2},{e^{2r_{1}}\over 2},\ldots ,{e^{-2r_{m}}\over 2},{e^{2r_{m}}\over 2}\right)W T_{S} -O^{T} W^{T}\text{diag}\left({e^{-2r_{1}}\over 2},{e^{2r_{1}}\over 2},\ldots ,{e^{-2r_{m}}\over 2},{e^{2r_{m}}\over 2}\right)W O \Vert_{2}^{2}\right] \nonumber \\
&= {1\over 4m(\log D)^{m}} \int d\vec{r}dW \Vert (T_{S}\otimes T_{S} - O\otimes O)(W\otimes W)\sum_{j=1}^{m}\left( e^{-2r_{j}}e_{2j-1}\otimes e_{2j-1} + e^{2r_{j}}e_{2j}\otimes e_{2j} \right)  \Vert^{2}
\end{align}
\end{footnotesize}
with $dW$ the normalized Haar measure over $\text{Orth}(2m)$, and the second equality follows from using the isometric property of ${\bm{vec}}$. In the last equality, we also shortened the integral notation. Expanding the square of the Euclidean distance gives a sum of two integrals:
\begin{align}
R_{O}(T_{S})&= {2\over 4m(\log D)^{m}}\int d\vec{r} \sum_{j=1}^{m}\left( e^{-4r_{j}}+e^{4r_{j}} \right)\nonumber \\
&- {2\over 4m(\log D)^{m}}\int d\vec{r} \int dW \sum_{k,j=1}^{m} \left[ e^{-2(r_{j}+r_{k})}(e_{2j-1}^{T}W^{T}T_{S}^{T}OWe_{2k-1})^{2} \right. \nonumber \\
&{} \left. + e^{-2(r_{j}-r_{k})}(e_{2j-1}^{T}W^{T}T_{S}^{T}OWe_{2k})^{2} \right. \nonumber \\
&{} \left. + e^{2(r_{j}-r_{k})}(e_{2j}^{T}W^{T}T_{S}^{T}OWe_{2k-1})^{2} \right. \nonumber \\
&{} \left. + e^{2(r_{j}+r_{k})}(e_{2j}^{T}W^{T}T_{S}^{T}OWe_{2k})^{2} \right] 
\label{eqn:yyy2}
\end{align}

The integral in the first line evaluates to ${D^{2}-D^{-2}\over 4 \log D}$.  In the second integral, it is useful to break up the sum over $j,k$ to the $j=k$ and $j\neq k$ parts. Then we use the following lemma

\begin{lemma}
Let $V$ be a real vector space with orthonormal basis $\lbrace e_{i}\rbrace_{i=1}^{2m}$ and let $W,L\in \mathrm{Orth}(2m)$. Then for any $i,j\in \lbrace 1,\ldots, 2m\rbrace$ with $j\neq i$
\begin{align}
\int dW (e_{i}^{T}W^{T}LWe_{i})^{2} &= {\left( 2m+ \text{\emph{Tr}}L^{2}+(\text{\emph{Tr}}L)^{2} \right)\over 4m^{2}+4m} \nonumber \\
\int  dW (e_{i}^{T}W^{T}LWe_{j})^{2} &= \sum_{r=1}^{2m}{L_{r,r}^{2}\over 4m^{2}+4m} + \sum_{r\neq r'}{L_{r,r}L_{r',r'}(2m+1)\over 4m(m+1)(2m-1)} \nonumber \\
&\rightarrow {(\text{\emph{Tr}}L)^{2}\over 4m^{2}+4m} \text{ as } m\rightarrow \infty
\label{eqn:l2}
\end{align}
\end{lemma}

The proof of the lemma involves integration over the orthogonal group with respect to the Haar measure \cite{braun}. We will apply the lemma with $L=T^{T}_{S}O$, using the first integral from the lemma exactly to evaluate the $j=k$ part of the second integral in (\ref{eqn:yyy2}) and using the second integral from the lemma in its asymptotic form to evaluate the $j\neq k$ part of the second integral in (\ref{eqn:yyy2}). The result for integration over $W$ is
\begin{equation}
R_{O}(T_{S})={D^{2}-D^{-2}\over 4 \log D} - {(D^{2}-D^{-2})(2m+ \text{Tr}T_{S}^{T}OT_{S}^{T}O + (\text{Tr}T_{S}^{T}O)^{2} ) \over 16m(m+1)\log D} - { (D-D^{-1})^{2}(\text{Tr}T_{S}^{T}O)^{2} \over 8m(\log D)^{2}}
\label{eqn:r1}
\end{equation}

The assumption that $T_{S}$ and $O$ agree on the training dataset $S$ is now taken into account. Recall that the training covariance matrices $\Sigma^{(j)}$ are associated with distinct directions in $\mathbb{R}^{2m}$. Therefore, we can write $T_{S}^{T}O=\mathbb{I}_{\vert S\vert}\oplus Y$ with $Y\in \text{Orth}((2m-\vert S\vert))$. Note that $\int dY\; \text{Tr}Y = 0$, $\int dY\; \text{Tr}Y^{2}=1$, and $\int dY\; (\text{Tr}Y)^{2}=1$. From this it follows that $E_{O}\left( (\text{Tr}T_{S}^{T}O)^{2}\right)=E_{Y}((\vert S\vert+\text{Tr}Y)^{2})=\vert S\vert^{2}+1$ and $E_{O}(\text{Tr}\left[ (T_{S}^{T}O)^{2}\right])=E_{Y}(\vert S\vert+Y^{2})=\vert S\vert+1$. Applying these to (\ref{eqn:r1}) gives

\begin{equation}
E_{O}(R_{O}(T))= {D^{2}-D^{-2}\over 4 \log D} - {(D^{2}-D^{-2})(2m+ \vert S\vert^{2}+\vert S\vert+2 ) \over 16m(m+1)\log D} - { (D-D^{-1})^{2}(\vert S\vert^{2}+1) \over 8m(\log D)^{2}} + \mathcal{O}(m^{-2})
\end{equation}
where the $\mathcal{O}(m^{-2})$ comes from the asymptotic in (\ref{eqn:l2}). Absorbing the remaining $\mathcal{O}(m^{-2})$ terms and carrying out the trivial average over the finite set $S$ results in (\ref{eqn:r2}).\end{proof}

\end{document}